\documentclass[a4paper,USenglish,cleveref,autoref,thm-restate]{lipics-v2021}

\hideLIPIcs

\usepackage{amsmath,amsthm,amssymb,thmtools,thm-restate,booktabs,graphicx,latexsym,xcolor,xspace}
\usepackage{float}
\usepackage{bm}
\usepackage{tikz}
\usepackage{mathrsfs}
\usepackage{hyperref}

\definecolor{blueLink}{rgb}{0,0.2,0.8}
\hypersetup{colorlinks,linkcolor=blueLink,urlcolor=blueLink,citecolor=blueLink}

\usepackage{stmaryrd}
\usepackage{todonotes}
\usetikzlibrary{positioning,calc,arrows.meta,shapes.geometric,fit,shadows,decorations.pathmorphing}
\usetikzlibrary{backgrounds,trees}
\pgfdeclarelayer{background}
\pgfsetlayers{background,main}
\usetikzlibrary{patterns, patterns.meta}
\definecolor{prediction}{HTML}{10598A}
\definecolor{correctnode}{HTML}{1F77B4}
\definecolor{byznode}{HTML}{FF7F0E}
\definecolor{darkgreen}{HTML}{097969}
\tikzset{
  cor/.style 2 args={
    draw=black, thick,
    fill=correctnode!30,
    ellipse, minimum width=#1cm, minimum height=#2cm,
    inner sep=0pt, outer sep=3pt
  }
}

\tikzset{
  byz/.style 2 args={
    draw=black, thick,
    fill=byznode!40,
    ellipse, minimum width=#1cm, minimum height=#2cm,
    inner sep=0pt, outer sep=3pt,
    decorate, decoration={zigzag, segment length=4pt, amplitude=0.5pt}
  }
}

\tikzset{
  normal/.style 2 args={
    draw=black, very thick,
    fill=black!10,
    ellipse, minimum width=#1cm, minimum height=#2cm,
    inner sep=0pt, outer sep=3pt
  }
}

\usepackage[ruled, lined, linesnumbered, commentsnumbered, noend, ]{algorithm2e}
\usepackage{subcaption}
\usepackage{multirow}

\newcommand{\alg}{\textsf{ALG}\xspace}
\newcommand{\authAlg}{Dolev-Strong }

\newcommand{\globalg}{\text{pred\_ba}}
\newcommand{\localg}{\text{auth\_pred\_ba}}

\title{
  Resilient Byzantine Agreement with Predictions
}

\author{Julien Dallot}{TU Berlin, Germany}{julien.dallot@tu-berlin.de}{}{}
\author{Darya Melnyk}{TU Berlin, Germany}{melnyk@tu-berlin.de}{}{}
\author{Tijana Milentijevic}{TU Berlin, Germany}{tijana.milentijevic@tu-berlin.de}{}{}
\author{Stefan Schmid}{TU Berlin and Weizenbaum Institute, Germany}{stefan.schmid@tu-berlin.de}{}{}
\author{Patrik Welters}{HU Berlin, Germany}{patrik.welters@student.hu-berlin.de}{}{}

\authorrunning{J. Dallot, D. Melnyk, T. Milentijevic, S. Schmid and P. Welters}



\funding{
This work was supported by the German Research Foundation (DFG), Schwerpunktprogramm SPP 2378: ReNO-2 (511099228), 2025-2029
}

\nolinenumbers

\begin{document}
\maketitle 

\begin{abstract}
  

  The Byzantine Agreement (BA) problem is a fundamental task in distributed computing where nodes in a network need to agree on a common output in the presence of arbitrary (worst-case) node failures. In real-world applications, nodes can be monitored over time to provide ML predictions of faulty behavior in a network. In this work, we answer the question of whether predictions can improve the fault tolerance of a distributed system without sacrificing correctness when the predictions are wrong.
  We consider a prediction-augmented study of Byzantine agreement in which each node receives, in addition to its input bit, a predicted set of honest nodes.
  We focus on \emph{algorithmic resilience} --- the maximum number of faulty nodes an algorithm can tolerate --- and present algorithms and impossibility results whose resilience depends on the accuracy of the predictor.
  As our first main result, we bring a complete characterization of the consistency--robustness trade-offs in both the non-authenticated and authenticated settings:
  for $n$ nodes and a parameter $\alpha \in [0, 1]$, we present algorithms that tolerate up to $\alpha \cdot n$ faulty nodes when the predictor is correct (\emph{consistency}), and up to $\frac{1-\alpha}{2} \cdot n - 1$ faulty nodes when the predictor is arbitrarily wrong (\emph{robustness)}. In the authenticated setting, the robustness bound improves to $(1-\alpha) \cdot n - 1$.
  We prove matching impossibility results, showing that these tradeoffs are optimal and independent of the particular prediction system.
  Our second main result characterizes \emph{smoothness}: the rate at which resilience degrades as the predictor becomes less accurate.
  We show that resilience linearly decreases in the number of wrong predictions as long as that number stays within a constant fraction of $n$.
  Concretely, in the non-authenticated setting, each additional wrong prediction loses one unit of resilience, whereas in the authenticated setting, the decline is halved, since two wrong predictions are needed to lose one unit of resilience.
\end{abstract}




\maketitle

\section{Introduction}
\label{sec:introduction}

Byzantine Agreement (BA)~\cite{TwoGenerals,PeaseShostackLamport} is a most fundamental and intensively studied problem in distributed computing. It requires nodes to agree on a common output despite arbitrary, worst-case node failures, called Byzantine behavior in the literature. 
This work explores how predictions about the other nodes' behavior can be used to improve the resilience of BA. Indeed, in practice, nodes can often monitor and collect data about other nodes, and use machine-learned classifiers, anomaly detectors, or reputation algorithms, to guess how trustworthy other nodes are. 
This is a natural model for many distributed systems operating over longer time periods, such as blockchain: suspicious behavior of the participants can be monitored and, to some extent, detected~\cite{cachin2017blockchain}. However, with new participants and new capabilities of these systems, mistakes in the prediction are not preventable. 

Motivated by Machine Learning (ML), prediction-augmented algorithms have already received much attention and have been applied successfully in various domains already. Early work originated in the context of online algorithms, where predictions about the future events are used by algorithms to overcome worst-case bounds~\cite{DBLP:conf/icml/LykourisV18}, and many significant insights have been obtained how to improve the competitive ratio of different online algorithms~\cite{gupta2022augmenting, eliavs2024learning, purohit2018improving, zhao2024learning, daneshvaramoli2025nearoptimal, DBLP:conf/icml/LykourisV18}. But also data structures~\cite{benomar2024learning} and most recently also distributed computing tasks~\cite{GAwithPredictions,BAwithPredictions} have been improved with predictions. 
In particular, Ben-David et al.~\cite{BAwithPredictions} studied whether nodes can benefit from predictions about the correctness of other nodes in Byzantine Agreement to speed up agreement, hence focusing on the \emph{efficiency}.

In this paper, we are interested in the \emph{resilience} of prediction-augmented algorithms. In particular, we initiate the study of 
whether predictions about node behaviors can allow us to tolerate more node failures in Byzantine Agreement. This is challenging since predictions are inherently unreliable. And while accurate predictions can substantially improve the resilience, relying on them unconditionally can compromise correctness when they are wrong. Hence, as with any prediction-augmented algorithm, a useful prediction-augmented Byzantine agreement protocol should navigate the trade-off between \emph{consistency}, i.e., benefit when predictions are correct, and \emph{robustness}, i.e., not suffer when predictions are of poor quality or even adversarial. Beyond these extremes, we aim for \emph{smoothness}: the protocol should lose resilience gradually, rather than abruptly, as the number of incorrect predictions increases.


\subsection{Contributions}
\label{sec:contributions}

This paper initiates the study of whether and how ML predictions can be used to improve the resilience of algorithms for the Byzantine Agreement problem.
Specifically, we ask: given $n$ nodes and $t$ faults, can we use predictions to overcome the upper bounds of $n>3t$ (without authentication)~\cite{PeaseShostackLamport}
and $n>2t$ (with authentication)~\cite{TwoGenerals} on the resilience of the agreement protocols? 
We answer these questions affirmatively. 

We design prediction-augmented algorithms which, given a prediction of which nodes are trustworthy (honest), surpass any prediction-free algorithm in case of accurate predictions (i.e., consistency) while ensuring provable resilience guarantees \emph{against any prediction} (i.e., robustness).
In addition, we study the resilience of our algorithms as the prediction accuracy declines (i.e., smoothness) and provide resilience guaranties for any fraction of accurately predicted nodes.
Our algorithms are tuned with a \emph{trust parameter} $\alpha \in [0, 1]$ which is the desired consistency (tolerated fraction of faulty nodes for good predictions) and interpolates between prudent or risky algorithms depending on how much one trusts the predictor.
We provide theoretical results for four settings: we consider two communication models, where the communication channels can either be authenticated or not; and we consider two models of predictions, a global model in which the predictions are publicly available to all nodes (e.g., broadcast by an ML monitoring system) and a local model in which predictions are only available locally. Our results for global predictions are summarized in \autoref{fig:smoothness-curves}.
Starting from existing, prediction-free agreement algorithms, we build prediction-augmented algorithms to achieve better resilience bounds in the non-authenticated setting (\autoref{sec:non-auth} and \autoref{sec:non-auth-impossibility}) as well as the authenticated setting (\autoref{sec:auth}) and finally derive impossibility results for local predictions (\autoref{sec:localPredictions}).


\noindent\textbf{Characterizing consistency and robustness.}
Our first technical contribution is a complete characterization of the possible consistency and robustness trade-offs under global predictions for non-authenticated and authenticated channels.
Results in the global setting are summarized in~\autoref{tab:consistency-robustness-summary}.
\begin{table}[t]
  \centering
  \setlength{\tabcolsep}{8pt}
  \renewcommand{\arraystretch}{1.4}
  \small
  \begin{tabular}{|l|l|c|c|c|}
    \hline
    Pred. & Communications & \textbf{Consistency} & \textbf{Robustness} & \textbf{Impossible Rob.} \\
    \hline
    \multirow{2}{*}{Global}
    & Non-Auth & $\alpha \in \left[\tfrac{1}{3}, 1\right]$ & $\tfrac{1-\alpha}{2} - \frac{1}{n} \ ^{(\text{Thm.} \ref{thm:consistency-robustness})}$ & $\frac{1-\alpha}{2} \ ^{(\text{Thm.} \ref{thm:impossibility-consistency-robustness})}$ \\[4pt]
    \cline{2-5}
    & Auth     & $\alpha \in \left[\tfrac{1}{2}, 1\right]$ & $(1-\alpha) - \frac{1}{n} \ ^{(\text{Thm.} \ref{thm:auth-consistency-robustness})}$           & $(1-\alpha) \ ^{(\text{Thm.} \ref{thm:auth-impossibility-consistency-robustness})}$           \\
    \hline
    \multirow{2}{*}{Local}
    & Non-Auth & $\alpha \in \left[\tfrac{1}{3}, \tfrac{1}{2}\right[$ & (unknown)                          & (unknown)                   \\
    \cline{2-5}
    & Non-Auth \& Auth     & $\alpha \in \left[\tfrac{1}{2}, 1\right]$ & $0$                                 & $\frac{1}{n} \ ^{(\text{Thm.} \ref{thm:impossibility-local})}$                          \\
    \hline
  \end{tabular}
  \vspace{8pt}
  \caption{
    Our results on consistency-robustness trade-offs.
    On each row, the first two columns indicate the setting, the third column shows the target consistencies that are relevant --- so it is superior to prediction-free algorithms ---, the fourth column shows our robustness for the target consistency $\alpha$, and the last column shows an impossible robustness for a consistency of $\alpha$.
  }
  \label{tab:consistency-robustness-summary}
\end{table}
We derive our tight characterization with indistinguishability arguments~\cite{TwoGenerals,PeaseShostackLamport}: we show that, given any algorithm with a certain consistency $\alpha$, we can construct a set of configurations and prediction inputs among which at least one defeats any algorithm with a robustness too high ($\frac{1-\alpha}{2}$ for the non-authenticated case, $1 - \alpha$ for the authenticated case).
These results show that trade-offs cannot be avoided when it comes to resilience: unlike round complexity~\cite{BAwithPredictions}, it is not possible to use predictions to strictly improve an algorithm in all cases, for instance by running a ``back-up'' algorithm that could take over if the prediction-augmented algorithm fails --- in our case, it is impossible to detect that the prediction-augmented algorithm has failed in the first place.

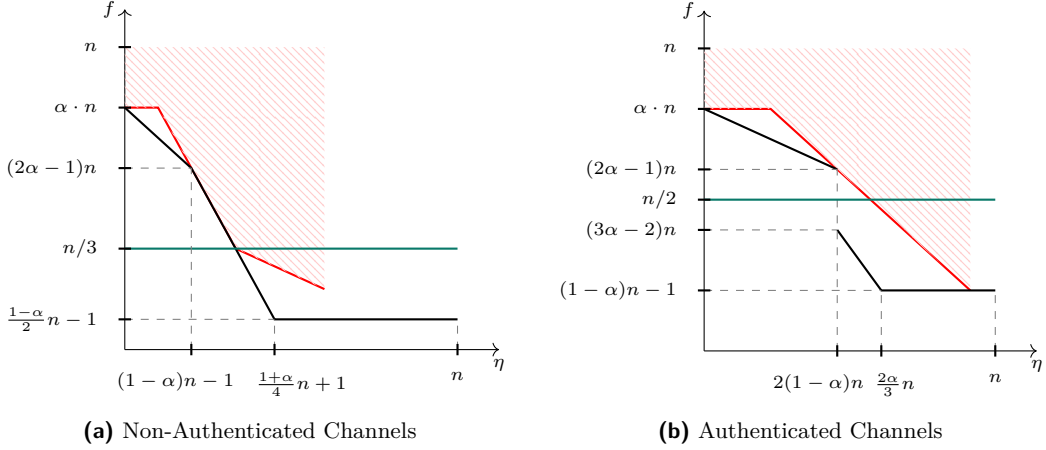
\begin{figure}[tb]
  \centering
  \begin{subfigure}[t]{0.48\textwidth}
  \centering
  \begin{tikzpicture}[xscale=0.55, yscale=0.5]

    \draw[->] (0, 0) -- (9, 0) node[below] {\scriptsize $\eta$};
    \draw[->] (0, 0) -- (0, 9) node[left] {\scriptsize $f$};

    \def\Fx{0}
    \def\Fy{8}
    \coordinate (F) at (\Fx, \Fy);

    \def\Gx{2.67}
    \def\Gy{2.67}
    \coordinate (G) at (\Gx, \Gy);

    \def\Hx{4.8}
    \def\Hy{1.6}
    \coordinate (H) at (\Hx, \Hy);

    \def\Ix{4.8}
    \def\Iy{8}
    \coordinate (I) at (\Ix, \Iy);

    \def\Jx{0.8}
    \def\Jy{6.4}
    \coordinate (J) at (\Jx, \Jy);

    \draw[thick, red] (0, \Jy) -- (J) -- (G) -- (H);

    \fill[fill=red, pattern color=red!20, pattern=north west lines] (0, \Jy) -- (J) -- (G) -- (H) -- (I) -- (F) -- cycle;

    \def\Ax{0}
    \def\Ay{6.4}
    \coordinate (A) at (\Ax, \Ay);

    \def\Bx{1.6}
    \def\By{4.8}
    \coordinate (B) at (\Bx, \By);

    \def\Cx{3.6}
    \def\Cy{0.8}
    \coordinate (C) at (\Cx, \Cy);

    \def\Dx{8}
    \def\Dy{0.8}
    \coordinate (D) at (\Dx, \Dy);

    \draw[thick] (A) -- (B) -- (C) -- (D);

    \draw[thick] (0 - 0.15, 8) -- (0 + 0.15, 8) node[left, xshift=-9pt] {\scriptsize $n$};
    \draw[thick] (\Ax - 0.15, \Ay) -- (\Ax + 0.15, \Ay) node[left, xshift=-9pt] {\scriptsize $\alpha \cdot n$};
    \draw[dashed, gray] (0, \By) -- (\Bx, \By);
    \draw[thick] (-0.15, \By) -- (0.15, \By) node[left, xshift=-9pt] {\scriptsize $(2 \alpha - 1) n$};
    \draw[dashed, gray] (0, \Cy) -- (\Cx, \Cy);
    \draw[thick] (-0.15, \Cy) -- (0.15, \Cy) node[left, xshift=-9pt] {\scriptsize $\frac{1-\alpha}{2} n - 1$};
    \draw[thick, darkgreen] (0, 2.67) -- (8, 2.67);
    \draw[thick] (-0.15, 2.67) -- (0.15, 2.67) node[left, xshift=-9pt] {\scriptsize $n / 3$};

    \draw[dashed, gray] (\Bx, 0) -- (\Bx, \By);
    \draw[thick] (\Bx, -0.15) -- (\Bx, 0.15) node[below, yshift=-7pt, xshift=-6pt] {\scriptsize $(1 - \alpha) n - 1$};
    \draw[dashed, gray] (\Cx, 0) -- (\Cx, \Cy);
    \draw[thick] (\Cx, -0.15) -- (\Cx, 0.15) node[below, yshift=-7pt, xshift=10pt] {\scriptsize $\frac{1 + \alpha}{4} n + 1$};
    \draw[dashed, gray] (\Dx, 0) -- (\Dx, \Dy);
    \draw[thick] (\Dx, -0.15) -- (\Dx, 0.15) node[below, yshift=-7pt] {\scriptsize $n$};

  \end{tikzpicture}
  \caption{
    \centering
    Non-Authenticated Channels
  }
  \label{fig:smoothness-curve}
  \end{subfigure}
  \hfill
  \begin{subfigure}[t]{0.48\textwidth}
  \centering
  \begin{tikzpicture}[xscale=0.55, yscale=0.5]

    \draw[->] (0, 0) -- (8, 0) node[below] {\scriptsize $\eta$};
    \draw[->] (0, 0) -- (0, 9) node[left] {\scriptsize $f$};

    \def\Fx{8*0.2}
    \def\Fy{8*0.8}
    \coordinate (F) at (\Fx, \Fy);

    \def\Gx{6.4}
    \def\Gy{1.6}
    \coordinate (G) at (\Gx, \Gy);

    \def\Hx{6.4}
    \def\Hy{8}
    \coordinate (H) at (\Hx, \Hy);

    \def\Ix{8*0.2}
    \def\Iy{8}
    \coordinate (I) at (\Ix, \Iy);

    \def\Ax{0}
    \def\Ay{6.4}
    \coordinate (A) at (\Ax, \Ay);

    \def\Jx{0}
    \def\Jy{8}
    \coordinate (J) at (\Jx, \Jy);

    \draw[thick, red] (A) -- (F) -- (G);

    \fill[fill=red, pattern color=red!20, pattern=north west lines] (A) -- (F) -- (G) -- (H) -- (I) -- (J) -- cycle;

    \def\Bx{3.2}
    \def\By{4.8}
    \coordinate (B) at (\Bx, \By);

    \def\Cx{3.2}
    \def\Cy{3.2}
    \coordinate (C) at (\Cx, \Cy);

    \def\Dx{4.26}
    \def\Dy{1.6}
    \coordinate (D) at (\Dx, \Dy);

    \def\Ex{7}
    \def\Ey{1.6}
    \coordinate (E) at (\Ex, \Ey);

    \draw[thick] (A) -- (B);
    \draw[thick] (C) -- (D) -- (E);

    \draw[thick] (0 - 0.15, 8) -- (0 + 0.15, 8) node[left, xshift=-9pt] {\scriptsize $n$};

    \draw[thick] (\Ax - 0.15, \Ay) -- (\Ax + 0.15, \Ay) node[left, xshift=-9pt] {\scriptsize $\alpha \cdot n$};

    \draw[dashed, gray] (\Bx, 0) -- (\Bx, \By);
    \draw[thick] (\Bx, -0.15) -- (\Bx, 0.15) node[below, yshift=-7pt, xshift=-7pt] {\scriptsize $2(1-\alpha) n$};
    \draw[dashed, gray] (0, \By) -- (\Bx, \By);
    \draw[thick] (-0.15, \By) -- (0.15, \By) node[left, xshift=-9pt] {\scriptsize $(2 \alpha - 1) n$};

    \draw[thick, darkgreen] (0, 4) -- (7, 4);
    \draw[thick] (-0.15, 4) -- (0.15, 4) node[left, xshift=-9pt] {\scriptsize $n / 2$};

    \draw[dashed, gray] (0, \Cy) -- (\Cx, \Cy);
    \draw[thick] (-0.15, \Cy) -- (0.15, \Cy) node[left, xshift=-9pt] {\scriptsize $(3 \alpha - 2) n$};

    \draw[dashed, gray] (\Dx, 0) -- (\Dx, \Dy);
    \draw[thick] (\Dx, -0.15) -- (\Dx, 0.15) node[below, yshift=-7pt, xshift=5pt] {\scriptsize $\frac{2 \alpha}{3} n$};
    \draw[dashed, gray] (0, \Dy) -- (\Dx, \Dy);
    \draw[thick] (-0.15, \Dy) -- (0.15, \Dy) node[left, xshift=-9pt] {\scriptsize $(1-\alpha) n - 1$};

    \draw[dashed, gray] (\Ex, 0) -- (\Ex, \Ey);
    \draw[thick] (\Ex, -0.15) -- (\Ex, 0.15) node[below, yshift=-7pt] {\scriptsize $n$};

  \end{tikzpicture}
  \caption{
    \centering
    Authenticated Channels
  }
  \label{fig:auth-smoothness-curve}
  \end{subfigure}
  \caption{
    This figure shows smoothness under global predictions for the non-authenticated and authenticated settings (obtained for $\alpha = 0.8$).
    In each plot, the black curve is the algorithm's smoothness, i.e., the resilience with respect to the error $\eta$, the red curve and the dashed red region are impossible regions, and the dark green horizontal line is the optimal resilience of a prediction-free algorithm.
    Observe that the resilience has no sudden jumps in the range of almost accurate predictions.
    The descending curves show great tolerance to prediction errors and allow for a clear supremacy of our algorithms compared to standard prediction-free ones on a wide range of imperfect predictions.
  }
  \label{fig:smoothness-curves}
\end{figure}

\noindent\textbf{Inaccurate predictions are still beneficial.} Our second main result of this paper is on the \emph{smoothness}~\cite{DBLP:conf/icml/LykourisV18} of our algorithms, i.e., how their resilience evolves when errors are gradually introduced in the prediction.
As noted by Boyar et al~\cite{GAwithPredictions}, this metric is of particular interest for distributed systems: when $n$ is large, the consistency metric loses practical value since perfectly predicting the behavior of every node becomes increasingly unlikely.
Following~\cite{BAwithPredictions}, we define an \emph{error function} $\eta \in \llbracket 0, n \rrbracket$.
Then, we derive resilience guarantees for our algorithms for each possible value of the error function in both non-authenticated and authenticated settings.
Our main result is a linear decrease of the resilience for good predictions: in the non-authenticated setting, each new error loses (in the worst case) one unit of resilience; the decline is halved in the authenticated setting as to new errors are needed to lose one unit of resilience.
Overall, our results show a clear supremacy of prediction-augmented algorithms compared to prediction-free ones, \emph{even for a wide range of imperfect predictions}; our results are summarized in~\autoref{fig:smoothness-curves}.


\noindent\textbf{Impossibility results for local predictions.} Our last contribution is a strong impossibility result in the case of local prediction, i.e., when each node can obtain its own input prediction (and therefore predictions can vary between two nodes).
We show that, in this local prediction framework, any algorithm with a consistency of $1/2$ must have a robustness of $0$. 
This result directly implies that local predictions do not provide better guarantees in the authenticated case, from a consistency and robustness perspective.
For the non-authenticated case, the following interesting question remains open for future work: Is it possible to obtain non-trivial consistency and robustness trade-offs for a consistency inside $[1/3, 1/2[$?

\section{Model}
\label{sec:model}

\noindent\textbf{Problem Setup.} We consider the fundamental problem of Byzantine agreement, where $n$ nodes in a distributed system are given a binary input value and need to agree on a common output value by communicating with each other.
Each node has a unique ID, let the IDs be $1, 2 \dots n$, for simplicity.
The underlying communication network is synchronous and fully connected.
We assume that the nodes all start the protocol at the same time.
We therefore can interpret the synchronous communication as communication in rounds, where in each round, a node can (1) send any sequence of bits to its neighbors, (2) receive messages from its neighbors, and (3) update its local state.

An (unknown) subset of nodes $F \subseteq \llbracket n \rrbracket$ are assumed to show (Byzantine) \emph{faults}, meaning they may deviate from the algorithm in an arbitrary (worst-case) manner.
Those faulty nodes may collude without limit to prevent the other nodes from reaching agreement.
We use $f$ to denote the number of faulty nodes in the system.
While $f$ is unknown, we assume that the nodes have access to a known upper-bound on $f$, called $t$; in this paper, we set $t$ so that the problem is simply impossible to solve if $f > t$.
We denote the subset of \emph{honest} nodes, i.e.,  non-faulty, by $H = \llbracket n \rrbracket \setminus F$.

\noindent\textbf{Authentication.} In this paper, we focus on two standard models, one with \emph{non-authenticated} communication channels and the other with \emph{authenticated} communication channels.
In the non-authenticated setting, the messages transmitted over the communication channels are in plain, non-encrypted text. 
A faulty node can therefore freely alter its communication history: claiming to have sent or received messages that never existed, denying actual messages, or presenting different histories to different nodes.
In the authenticated setting (PKI model) however, nodes can digitally sign their messages, and the adversary is assumed computationally bounded, making it infeasible to forge or alter signed messages.
The adversary can still read each message, but it cannot fabricate a false communication history.

\noindent\textbf{Prediction Model.} In addition to their standard input values, this paper assumes the nodes receive an additional input, the \emph{prediction} --- provided by, e.g., a monitoring system that detects suspicious behaviors --- indicating which nodes are honest.

\begin{definition}[Prediction]
  \label{def:prediction}
  A prediction is a set $P \subseteq \llbracket n \rrbracket$ which indicates the honest nodes in the system.
  Predictions are unreliable and may arbitrarily differ from the actual set of honest nodes $H$.
\end{definition}

We consider two prediction models: \emph{global} and \emph{local}.
In the \emph{global model}, all nodes receive the same prediction $P$ (in that case, we therefore omit the node subscript).
This captures, e.g., a centralized monitoring system that flags nodes for suspicious behaviors and broadcasts its assessment of which nodes are trustworthy.
In the \emph{local model}, each node receives its own prediction, computed independently; the predictions may differ arbitrarily across nodes (same model as in~\cite{BAwithPredictions}).

For the most part of this paper (\autoref{sec:non-auth} and \autoref{sec:auth}), we focus on the global model; the definitions below therefore concern global predictions only.
We present impossibility results for local predictions in \autoref{sec:localPredictions} where we define analogous definitions for local predictions.

Our goal is to design synchronous algorithms that use predictions to solve the \emph{Byzantine agreement problem}.
We now formally define such algorithms.

\begin{definition}[prediction-augmented Byzantine Agreement Problem]
  \label{def:pred-algo}
  Each node $i$ receives two inputs: (1) a \emph{binary value} $v_i$ and (2) a \emph{prediction} $P_i$.
  A prediction-augmented algorithm \alg solves the Byzantine agreement problem if it satisfies:
  \begin{itemize}
  \item
    \textbf{Termination:} Every honest node terminates in finite time.
  \item
    \textbf{Agreement:} All honest nodes terminate with the same output bit.
  \item
    \textbf{Validity:} The output of every honest node is the input value of some honest node.
  \end{itemize}
\end{definition}



\noindent\textbf{Evaluation Metrics.} In this paper, we evaluate algorithms using the standard \emph{resilience} metric: the maximum number of faulty nodes the algorithm can tolerate.
Without predictions, the optimal resilience bounds are well-known: $\lceil n/3\rceil - 1$ faulty nodes without authentication, and $\lceil n/2\rceil - 1$ with authentication.

We study how these bounds improve when nodes have access to predictions.
Following the prediction-augmented algorithms literature~\cite{DBLP:conf/icml/LykourisV18}, we use three metrics defined hereafter: \emph{consistency}, \emph{robustness}, and \emph{smoothness}.
Consistency measures resilience under correct predictions:

\begin{definition}[Consistency]
  \label{def:consistency}
  Let $\alpha \in [0, 1]$.
  A prediction-augmented algorithm \alg is $\alpha$-consistent if, for every configuration with at most $\alpha \cdot n$ faulty nodes, there exists a prediction assignment under which \alg reaches agreement.
\end{definition}

Robustness measures resilience regardless of the predictions:

\begin{definition}[Robustness]
  \label{def:robustness}
  Let $\beta \in [0, 1]$.
  A prediction-augmented algorithm \alg is $\beta$-robust if, for every configuration with at most $\beta \cdot n$ faulty nodes and every prediction assignment, \alg reaches agreement.
\end{definition}

Consistency and robustness are two extremes: the former assumes perfect predictions, the latter assumes nothing about them.
Smoothness interpolates between the two, capturing how resilience degrades as prediction quality worsens.
Smoothness is parameterized by an error function $\eta$ that measures the quality of the prediction.

\begin{definition}[Error Function]
  \label{def:error-function}
  Let $H$ be the set of honest nodes and $P$ a global prediction.
  The \emph{error} of $P$ is:
  \begin{align*}
    \eta(H, P) = |H \setminus P| + |P \setminus H|
  \end{align*}
  That is, $\eta(H, P)$ counts all mispredictions: nodes in $H$ but not $P$ (honest but wrongly predicted faulty) and nodes in $P$ but not $H$ (faulty but predicted honest).
\end{definition}

\begin{definition}[Smoothness]
  \label{def:smoothness}
  The \emph{smoothness} of algorithm \alg is the function $s : \llbracket 0, n \rrbracket \to \mathbb{N}$ where $s(i)$ is the maximum number of faulty nodes \alg can tolerate when the prediction error satisfies $\eta(H, P) = i$.
\end{definition}


\noindent\textbf{Notations.} Let $i$ and $j$ be two integers with $i \le j$.
By $\llbracket i, j \rrbracket$ we refer to the set $\{i, i+1 \dots j-1, j\}$, by $\llbracket i \rrbracket$ we refer to the set $\{1, 2 \dots i\}$.
Let $i \in \mathbb{R}$ be a real number, we write $\lfloor i \rfloor$ to refer to the \emph{inferior integer part} of $i$, that is, the largest integer number smaller than $i$; likewise, we write $\lceil i \rceil$ to refer to the \emph{superior integer part} of $i$, which is the smallest integer number larger than $i$.
For convenience, we extend the range notations to real numbers: let $i, j \in \mathbb{R}$ with $i \le j$, $\llbracket i, j \rrbracket$ refers to the set $\{\lceil i \rceil, \lceil i \rceil + 1 \dots \lfloor j \rfloor -1, \lfloor j \rfloor\}$ and $\llbracket i \rrbracket$ refers to the set $\{1, 2 \dots \lfloor i \rfloor -1, \lfloor i \rfloor\}$.

\section{Reaching Agreement with Non-Authenticated Channels}
\label{sec:non-auth}

This section introduces our prediction-augmented algorithm $\globalg$ (\autoref{alg:ba}) for non-authenticated channels and global predictions.
We prove that our prediction-augmented algorithms admits a trade-off between consistency and robustness dependent on an input parameter $\alpha \in [1/3, 1]$ (\autoref{thm:consistency-robustness}) and establish the smoothness curve of $\globalg$ as shown in~\autoref{fig:smoothness-curve}  and \autoref{thm:smoothness-global-predictions}.

\begin{algorithm}[tbh]
  \caption{$\globalg$($v_x$, $P$, $\alpha$)}\label{alg:ba}
  \tcp{This algorithm runs in node $x$}
  \KwIn{the input bit $v_x$, prediction set $P \subseteq \llbracket n \rrbracket$, trust parameter $\alpha \in [1/3,1]$}

  \tcp{Initialize $L$ to the predicted set}
  $L \gets P$\\
  \tcp{Add nodes to $L$ if needed}
  $y \gets 1$\\
  \If{$|L| < \frac{3}{2} \cdot (1-\alpha) \cdot n - 1$}{
    \Repeat{$|L| \ge \frac{3}{2} \cdot (1-\alpha) \cdot n - 1$}{
      \If{$y \notin L$}{
        $L \gets L \cup \{y\}$\\
      }
      $y \gets y + 1$\\
    }
  }
  \tcp{Run the agreement algorithm}
  $t \gets \left\lceil|L| / 3 \right\rceil$\\
  \eIf{$x \in L$}{
    Run the Phase King algorithm with input bit $v_x$ and with the other nodes in $L$\\
    Once an output value was found, broadcast it\\
    \KwOut{the output of the Phase King algorithm}
  }
  {
    Wait the number of rounds taken by the Phase King algorithm with $|L|$ nodes and $t-1$ faulty nodes to terminate\\
    \KwOut{a value received at least $|L| - t + 1$ times, otherwise output the initial input bit}
  }
\end{algorithm}

$\globalg$ (\autoref{alg:ba}) proceeds in two phases.
In the first phase, all honest nodes agree on a common subset $L \subseteq \llbracket n \rrbracket$ based on the prediction.
In the second phase, they run a standard Byzantine agreement algorithm restricted to active nodes in $L$; nodes outside $L$ are passive and simply adopt the decided value.
Our prediction-augmented algorithm uses the classic Phase King protocol from~\cite{PhaseKing}, note that any algorithm that also reaches agreement against $\lceil n / 3\rceil - 1$ faulty nodes or fewer can be used instead.

\noindent\textbf{Controlling the number of active nodes.} Ideally, $L$ would equal the set of honest nodes $H$, so each node initializes $L = P$.
The strategy that relies solely on $P$ is not robust: if $P$ contains too many faulty nodes, agreement cannot be guaranteed.
To address this, we enforce that $L$ has a minimum size of $\lceil 3(1-\alpha)/2 \cdot n \rceil$: if $P$ does not already satisfy this constraint, the algorithm deterministically adds nodes into $L$ until it reached the minimum size.
This simple lower-bound on the number of active nodes guarantees a trade-off between consistency and robustness:

\begin{restatable}[Consistency-Robustness]{theorem}{thmConsistencyRobustness}
  \label{thm:consistency-robustness}
  $\globalg$ (\autoref{alg:ba}) is $\alpha$-consistent and $\left(\frac{1-\alpha}{2} - 1/n\right)$-robust.
\end{restatable}

\begin{proof}
  We first prove the claim on consistency: we assume that the number of faulty nodes $f$ is lower or equal to $\alpha \cdot n$ and that $P = H$.
  Since $|H| \ge (1-\alpha) \cdot n$, we deduce that \autoref{alg:ba} inserts at most $\lceil \frac{3}{2} \cdot (1 - \alpha) \cdot n \rceil - 1 - \lceil (1-\alpha) \cdot n \rceil$ nodes into $P$ to obtain $L$.
  We obtain:
  \begin{align*}
    \left\lceil \frac{3}{2} \cdot (1 - \alpha) \cdot n \right\rceil - 1 - \left\lceil (1-\alpha) \cdot n \right\rceil
    \quad < \quad  \frac{3}{2} \cdot (1 - \alpha) \cdot n - (1 - \alpha) \cdot n
    \quad = \quad \frac{1}{2} \cdot (1 - \alpha) \cdot n
  \end{align*}
  where the first inequality holds since, for all real number $r$, we have $\lceil r \rceil - 1 < r$.
  This implies that $L$ contains strictly less than a third of faulty nodes.
  It is known that the Phase King algorithm reaches agreement when strictly fewer than $\lceil |L| / 3 \rceil$ nodes are faulty and the claimed consistency follows.

  We now prove the claimed robustness.
  We assume that the number of faulty nodes $f$ is lower or equal to $\lfloor \frac{1}{2} \cdot (1 - \alpha) \cdot n\rfloor - 1$, the predicted set $P$ is arbitrary.
  We now use that $|L| \ge \frac{3}{2} \cdot (1 - \alpha) \cdot n - 1$ and lower-bound the fraction of faulty nodes in $L$:
  \begin{align*}
    |F \cap L| \ \le \ f \ \le \ \left\lfloor \frac{1}{2} \cdot (1 - \alpha) \cdot n\right\rfloor - 1 \ \le \ \frac{1}{3} \cdot \left\lfloor \frac{3}{2} \cdot (1 - \alpha) \cdot n\right\rfloor - 1 \ \le \ \frac{|L|}{3} - \frac{2}{3} \ < \ \frac{|L|}{3}.
  \end{align*}
  This proves that the Phase King algorithm reaches agreement in $L$ and the claim follows.
\end{proof}

The consistency measure implies quite an extreme requirement on the predictor's accuracy and may thus rarely apply, especially when the number of nodes is large and a perfect identification of each node is a challenging task.
We now answer the question: How resilient is $\globalg$ when the predictor is \emph{almost} accurate, say, $99 \%$ of the nodes are correctly identified?
We answer this question in general, for any fraction of the nodes, by deriving a smoothness curve:

\begin{restatable}[Smoothness]{theorem}{thmSmoothnessGlobalPredictions}
  \label{thm:smoothness-global-predictions}
  Let $\alpha \in [1/3, 1]$.
  We consider the function $s$ defined as follows.
  \[
    s(\eta) = \begin{cases}
      \alpha \cdot n - \eta& \text{if } \eta \in \llbracket 0, (1-\alpha) \cdot n - 1\rrbracket, \\[6pt]
      n - 2\eta - 1& \text{if } \eta \in \left\llbracket (1-\alpha) \cdot n - 1, \frac{1+\alpha}{4} \cdot n + 1 \right\rrbracket, \\[6pt]
      \frac{1-\alpha}{2} \cdot n - 1 &\text{if } \eta \in \left\llbracket \frac{1+\alpha}{4} \cdot n + 1, n \right\rrbracket
    \end{cases}
  \]
  Assuming an error of $\eta$, $\globalg$ (\autoref{alg:ba}) with trust parameter $\alpha$ reaches agreement against $s(\eta)$ faulty nodes or fewer.
\end{restatable}

\begin{proof}[Proof Idea]
  In the following, we present the general proof idea. The full proof of this theorem can be found in Section \ref{sec:proof-smoothness-global-predictions}.

\noindent\textbf{No resilience drop.} Interestingly, there is no sudden drop in the resilience as the error increases, neither a critical quantity of mispredictions triggering a drastic phase transition. 
Instead, the curve is smooth: as long as the number of errors remains within a constant fraction of the nodes (from $0$ to $(1-\alpha) \cdot n - 1$), each new misprediction incurs at worst the loss of one unit of resilience.
To explain this, we position ourselves on a point of the smoothness curve somewhere between $\eta = 0$ to $\eta = (1-\alpha) \cdot n - 2$.
Let $C$ and $P$ be a configuration and a prediction that fit our current error.
Now make one of the predicted nodes faulty.
Notice that doing this has two simultaneous effects: both the error $\eta$ and the number of faulty nodes $f$ were incremented by one.
Clearly, $\globalg$ now cannot guarantee the same resilience as before: the faulty node we added into $P$ will also be an active node in $L$ and, as $\globalg$ was already at its resilience limit before, the faulty nodes are now numerous enough to prevent agreement.
But what happens if we slightly reduce the number of faulty nodes?
We now turn a faulty node into a honest node.
By doing so, \emph{the error must change}: if the node is in $P$ then $\eta$ was decremented by one, otherwise $\eta$ is incremented.
The key argument of the proof is that, in order to keep a constant error, \emph{not only the faulty nodes must change, but also the prediction}; in our case, we must first convert the faulty node (outside of $P$) into a honest node, and then insert that node inside of $P$.
Although it is now accurately predicted, this new honest node is not enough to balance our initial insertion of a faulty node; this is because, in the non-authenticated case, one requires at least two honest nodes for each faulty node to reach agreement.
Thus, we turn one other faulty node into a honest node (decrementing $f$ by one overall) which now guarantees agreement.

\noindent\textbf{The phase transition.} Interestingly, the smoothness curve is not simply linear.
Past an error of $\eta = (1-\alpha) \cdot n - 1$, the smoothness curve becomes steeper: in the worst case, each new misprediction now incurs the loss of two units of resilience.
Our impossibility result on smoothness~\autoref{thm:impossibility-smoothness-appendix} shows that this transition is not only \emph{necessary} (at the transition point, our smoothness curve touches the impossibility zone) but is also \emph{optimal} for a constant fraction of the errors (see~\autoref{fig:smoothness-curve}).
The transition occurs when mispredictions are numerous enough to consummate the initial advantage that comes with perfect predictions.
At the consistency point ($\eta = 0$), every predicted node is honest.
Then as the error grows, these correctly predicted honest nodes are replaced by faulty ones (in the worst case), eroding that advantage.
By the transition point, the first phase (slope $-1$) has consumed this advantage entirely: all nodes that were originally honest and well-predicted may now be faulty.
Beyond this point, the honest-to-faulty ratio in the \emph{predicted set} $P$ is itself at the edge of what agreement requires --- not only in the active set $L$.
Since $P$ no longer has enough honest nodes to absorb an additional faulty insertion from $L \setminus P$, one extra faulty node must be removed to restore the balance, which explains the steeper slope of $-2$.\qedhere
\end{proof}

\section{Impossibility Results with Non-Authenticated Channels}
\label{sec:non-auth-impossibility}
This sections shows impossibility results for any prediction-augmented algorithm: we prove that our algorithm $\globalg$ reaches the optimal consistency-robustness trade-off (\autoref{thm:impossibility-consistency-robustness}) and derive impossibility results for the smoothness metric (\autoref{thm:impossibility-smoothness-appendix}).
\begin{restatable}[Impossibility on Consistency and Robustness]{theorem}{thmImpossibilityConsistencyRobustness}
  \label{thm:impossibility-consistency-robustness}
  Let $\alpha \in [0,1[$.
  Any prediction-augmented algorithm for non-authenticated channels with a consistency greater or equal to $\alpha$ has a robustness strictly lower than $\frac{1-\alpha}{2}$.
  This result is independent from the chosen prediction system.
\end{restatable}

\begin{proof}[Proof idea.]
We derive this impossibility result with an \emph{indistinguishability argument}~\cite{PhaseKing}: assuming an algorithm is $\alpha$-consistent, we derive three configurations and a strategy from the faulty nodes so that, in at least one configuration, either, validity, agreement (see \autoref{def:pred-algo}) or robustness requirements do not hold. The full proof of this theorem can be found in Section~\ref{sec:proof-impossibility-consistency-robustness}. \qedhere

\end{proof}


The smoothness metric is  inherently interesting because prediction systems are not perfect and will mostly act between the bounds of consistency and robustness. 
We here show variable impossibility results that limit any algorithm that satisfies optimal consistency robustness tradeoff.

\pgfmathsetmacro{\coeff}{0.66}
\begin{figure}
  \begin{subfigure}{0.32\textwidth}
    \centering
    \begin{tikzpicture}[
      edge/.style={->, draw=black, thick, -Stealth},
      scale=0.7
      ]
      
      \node[cor={\coeff*2.0}{\coeff*1.2}] (A) at (-1.2, 2) {\small $\bm{(A)}$ : $\bm{0}$};
      \node[cor={\coeff*2.0}{\coeff*1.2}] (B) at (1.2, 2) {\small $\bm{(B)}$ : $\bm{1}$};
      \node[byz={\coeff*2.0}{\coeff*1.2}] (C) at (0, 0.5) {\small $\bm{(C)}$};
      \node[byz={\coeff*1.5}{\coeff*1.2}] (D) at (0, 4) {\small $\bm{(D)}$};

      \draw[edge, bend right=30] (C.east) to node [below right] {$\bm{1}$} (node cs:name=B, angle=-50);
      \draw[edge, bend left=30] (C.west) to node [below left] {$\bm{0}$} (node cs:name=A, angle=-130);

      \draw[edge, bend right=30] (D.west) to node [above left] {$\bm{0}$} (node cs:name=A, angle=130);
      \draw[edge, bend left=30] (D.east) to node [above right] {$\bm{1}$} (node cs:name=B, angle=50);

    \end{tikzpicture}
    \caption{\centering}
    \label{fig:1-a}
  \end{subfigure}
  \begin{subfigure}{0.32\textwidth}
    \centering
    \begin{tikzpicture}[
      edge/.style={->, draw=black, thick, -Stealth},
      scale=0.7
      ]
      
      \node[cor={\coeff*2.0}{\coeff*1.2}] (A) at (-1.2, 2) {\small $\bm{(A)}$ : $\bm{0}$};
      \node[byz={\coeff*2.0}{\coeff*1.2}] (B) at (1.2, 2) {\small $\bm{(B)}$ : $\bm{1}$};
      \node[cor={\coeff*2.0}{\coeff*1.2}] (C) at (0, 0.5) {\small $\bm{(C)}$ : $\bm{0}$};
      \node[byz={\coeff*1.5}{\coeff*1.2}] (D) at (0, 4) {\small $\bm{(D)}$};

      \draw[edge, bend right=30] (D.west) to node [above left] {$\bm{0}$} (node cs:name=A, angle=130);
      \draw[edge, bend left=30] (D.east) to node [above right] {$\bm{1}$} (node cs:name=B, angle=50);
    \end{tikzpicture}
    \caption{\centering}
    \label{fig:1-b}
  \end{subfigure}
  \begin{subfigure}{0.32\textwidth}
    \centering
    \begin{tikzpicture}[
      edge/.style={->, draw=black, thick, -Stealth},
      scale=0.7
      ]
      
      \node[byz={\coeff*2.0}{\coeff*1.2}] (A) at (-1.2, 2) {\small $\bm{(A)}$ : $\bm{0}$};
      \node[cor={\coeff*2.0}{\coeff*1.2}] (B) at (1.2, 2) {\small $\bm{(B)}$ : $\bm{1}$};
      \node[cor={\coeff*2.0}{\coeff*1.2}] (C) at (0, 0.5) {\small $\bm{(C)}$ : $\bm{1}$};
      \node[byz={\coeff*1.5}{\coeff*1.2}] (D) at (0, 4) {\small $\bm{(D)}$};

      \draw[edge, bend right=30] (D.west) to node [above left] {$\bm{0}$} (node cs:name=A, angle=130);
      \draw[edge, bend left=30] (D.east) to node [above right] {$\bm{1}$} (node cs:name=B, angle=50);
    \end{tikzpicture}
    \caption{\centering}
    \label{fig:1-c}
  \end{subfigure}
  \caption{
    In each configuration the nodes are given the same global prediction $P = A \cup B \cup C$.
    A missing arrow between a faulty and a honest node means, the behavior of the faulty node towards that honest node is not of importance for the proof. The sizes of the subsets are the following: $|A|=|B|=|C|= \eta$ and $|D|=n - 3 \cdot\eta$
  }
  \label{fig:smoothness-configs}
\end{figure}
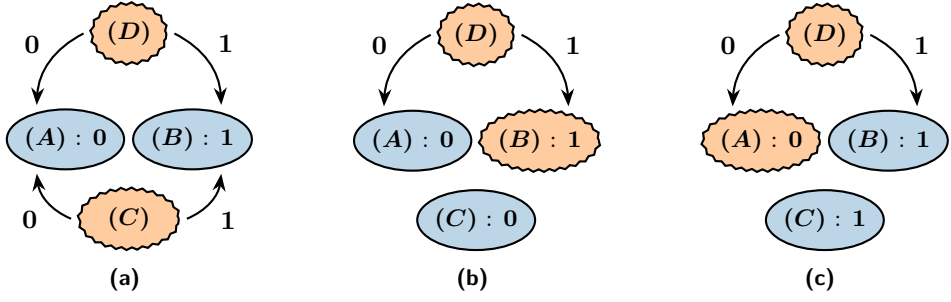

\begin{restatable}[Impossibility on Smoothness]{theorem}{thmImpossibilitySmoothnessAppendix}
	\label{thm:impossibility-smoothness-appendix}
	Let $\alpha \in [1/3, 1]$, let \alg be an algorithm that is $\alpha$-consistent.
	We consider the function $\bar{s}$ defined as follows.
	\[
		\bar{s}(\eta) = \left\{\begin{array}{l@{\quad}l@{\quad}r}
			\alpha \cdot n + 1                                 & if \eta \in \llbracket 0, \tfrac{(1-\alpha)}{2}\cdot n \rrbracket,                            & (1) \\[4pt]
			n - 2 \cdot \eta                                   & \text{if } \eta \in \llbracket \tfrac{(1-\alpha)}{2}\cdot n, \tfrac{1}{3} \cdot n\rrbracket, & (2) \\[4pt]
			\tfrac{1}{2} \cdot n - \tfrac{1}{2} \cdot \eta - 2 & \text{if } \eta \in \llbracket \tfrac{1}{3} \cdot n, \alpha \cdot n \rrbracket     & (3)
		\end{array}\right.
	\]
	The smoothness function of \alg is strictly lower than the piece $(1)$ and $(2)$ of $\bar{s}$; moreover, if the smoothness of \alg matches $(2)$ for $\eta \in \llbracket (1-\alpha)\cdot n, \frac{n}{3} \rrbracket$, then its smoothness cannot be above $(3)$.
\end{restatable}

\begin{proof}[Proof idea.] In the following, we present a proof idea for this theorem. The full proof is presented in Section~\ref{sec:proof-impossibility-smoothness-appendix}.
The proof for each interval of $\bar{s}(\eta)$ works by creating indistinguishable initial configurations.
\autoref{fig:smoothness-configs} shows the specific configurations for interval (2) in \autoref{thm:impossibility-smoothness-appendix}.
The figure depicts different subsets of nodes as ovals.
Orange ovals with curly borders depict faulty nodes and blue ovals depict honest nodes.
The honest nodes further have their input bit written in the oval, whereas the arrows originating from faulty subsets describe the strategy of the faulty nodes. In this case they behave as if they have the specific input bit indicated on the arrow.

The proof for interval (2) works by creating two opposite configurations (b) and (c), where subsets $A$ and $B$ must output $0$ and $1$ respectively.
The impossibility is then created in configuration (a).
If the Byzantine nodes in all configurations have a common strategy, they can make subset $A$ believe configuration (a) and (b) are the same.
The same holds for subset $B$ with configurations (a) and (c).
Because of the indistinguishability subset $A$ must output $0$ and subset $B$ must output $1$ in (a).
This breaks the agreement condition.
From this we can see that at least one of the configurations must not work.
By carefully choosing the sizes of the subsets and the prediction set, we can create the configurations such that each of them has the same error $\eta$ and number of faults $f$.
We conclude that no algorithm can solve Byzantine agreement with this resilience for the specific error.
\end{proof}

\section{Byzantine Agreement with Local Predictions}
\label{sec:localPredictions}
So far, we assumed that the prediction is global, i.e., all nodes receive the same predicted set of honest nodes. However, in practice in decentralized applications, different information might be available to different nodes. Consider a monitoring tool in which predictions are produced from locally observed traffic. In this case, nodes would obtain different views of the system and compute different predictions. In the setting with local predictions, the prediction is no longer a single global set, but a collection of sets, one per node.

\begin{definition}[Local Prediction]
  \label{def:local-prediction}
  A local prediction is a vector $\mathbf{P} = (P_1, \dots, P_n)$ where for every node $i \in \llbracket n \rrbracket$, $P_i \subseteq \llbracket n \rrbracket$ is the prediction received by node $i$.
  The interpretation is that node $i$ receives $P_i$ as input and treats it as the predicted set of honest nodes.
\end{definition}

In contrast to the global case, two honest nodes may now receive different predictions and node do not know the predictions received by others.
Hence, we update the error function $\eta$.

\begin{definition}[Error Function for Local Predictions]
  \label{def:local-error}
  Let $H$ be the set of honest nodes in configuration $C$, and let $\mathbf{P} = (P_1, \dots, P_n)$ be a set of local predictions.
  The \emph{prediction error} is 
  \begin{align*}
    \eta(C,\mathbf{P}) = \sum_{i\in H} (|P_i \setminus H| + |H \setminus P_i|)
  \end{align*}
\end{definition}
That is, we measure the total number of mispredicted nodes across all honest nodes. Note that in the worst case, the error function is in order of $\Theta(n^2)$.


Next, we present the main impossibility result for local predictions. This result shows that, unlike for global predictions, there is little room to improve the resilience of algorithms under local predictions.

\begin{restatable}{theorem}{thmImpossibilityLocal}
  \label{thm:impossibility-local}
  Let $\alpha \in [1/2, 1]$.
  Under local predictions, any algorithm that is $\alpha$-consistent is $0$-robust.
  This result holds for both non-authenticated and authenticated frameworks and any prediction system.
\end{restatable}

\begin{proof}[Proof sketch]
  (Refer to the full proof in \autoref{sec:proof-impossibility-local}).
  The proof uses the fact that nodes can receive different predictions to separate the nodes into two distinct groups, $A$ and $B$.
  We first consider the following two configurations: (1) $A$ is honest with input $0$ while $B$ is faulty and pretends to have input $1$, and (2) $A$ is faulty and pretends to have input $0$ while $B$ is honest with input $1$.
  Assume that the algorithm is $\alpha$-consistent, with $\alpha \ge \frac{1}{2}$: hence, there exists two local predictions, $P_1$ and $P_2$ such that, $A$ agree on $0$ when given $P_1$ in configuration (1) and $B$ agree on $1$ when given $P_2$ in configuration (2) (both by validity).
  Now, consider the following configuration (3): all nodes are honest, $A$ has input $0$ and $B$ has input $1$.
  Give $P_1$ to $A$ and give $P_2$ to $B$.
  $A$ cannot distinguish between (1) and (3); similarly, $B$ cannot distinguish between (2) and (3).
  Hence, $A$ and $B$ respectively output $0$ and $1$ in configuration (3), contradicting agreement.
\end{proof}

\section{Related Work}
\label{sec:related-work}

\noindent\textbf{Algorithms with predictions.}
Introduced by Lykouris and Vassilvitskii~\cite{DBLP:conf/icml/LykourisV18} for online caching, the framework of algorithms with predictions augments algorithms with an additional input — the \emph{prediction} — that may carry useful information about the problem instance.
No assumption is made on its quality: it can be perfect, adversarial, or anything in between.
The framework has since been applied to data structures, mechanism design, and distributed algorithms (see~\cite{OnlineAlgorithmsPredictionsWebsite} for an overview), and recently entered distributed computing with the seminal work of Boyar, Ellen, and Larsen~\cite{GAwithPredictions}.

\noindent\textbf{Byzantine agreement.}
Byzantine agreement was introduced in~\cite{TwoGenerals,PeaseShostackLamport}, establishing that up to $t < n/3$ failures can be tolerated without authentication, and up to $t < n/2$ with it~\cite{10.1145/28395.28420}.
For synchronous communication, $t+1$ rounds are necessary in the non-authenticated setting~\cite{FISCHER1982183}, a bound later extended to the authenticated case~\cite{DolevStrong}.
For asynchronous communication, the FLP impossibility result shows that even a single crash failure cannot be tolerated by a deterministic algorithm~\cite{FLPimpossibility}.
Our algorithms build on the Phase King protocol~\cite{PhaseKing}, though any synchronous Byzantine agreement protocol could serve as a drop-in replacement.
Byzantine agreement with Predictions was first proposed by Ben-David, Dzulfikar, Ellen and Gilbert~\cite{BAwithPredictions}. In this work, the authors present how classification predictions can help develop more efficient Byzantine agreement protocols. Note that, unlike in this work, the agreement algorithms in~\cite{BAwithPredictions} do not provide any robustness guarantees on the resilience of the protocols.


\noindent\textbf{Detection of Byzantine failures.} 
How the detection of faults can help us design more efficient agreement algorithms has been studied under different failure detection mechanisms. Failure detectors have been introduced by Chandra and Toueg~\cite{failureDetectorsChandraToueg} to solve consensus with crash failures. The idea is that failure detectors may be unreliable at first, but provide a perfect prediction after some amount of time. Malkhi and Reiter~\cite{10.5555/794197.795085} were the first to extend the definition of failure detectors to agreement with Byzantine failures. This work has been later extended in~\cite{8129490, BALDONI2003185}.
Another type of detecting Byzantine behavior is via accountability~\cite{10.5555/1973400.1973403}. The notion of accountability has been introduced for dependable distributed systems with the idea to use a scoring system that incentivizes honest behavior and punishes deviations from this behavior. Haeberlen, Kouznetsov, and Druschel~\cite{DBLP:conf/sosp/HaeberlenKD07} presented PeerReview, a system that provides accountability for distributed systems with Byzantine failures. In~\cite{10.5555/1973416.1973421}, several reasons are provided for implementing fault detection in distributed systems.


\section{Limitations and Future Work}
\label{sub:limitations}

Our work considers synchronous binary Byzantine Agreement with a fixed set of $n$ nodes and an upper bound $t$ on the number of Byzantine faults. Proposed algorithms focus on resilience, not efficiency: their runtime and communication complexity depend on the underlying Byzantine agreement subroutines.

An interesting future direction would be to find algorithms using local predictions with $\alpha>1/3$ and $\beta>0$. Based on this, the following open questions arise: Can one design protocols whose round complexity degrades smoothly with prediction error, while their resilience also degrades smoothly?
Can prediction accuracy be leveraged to simultaneously approach optimal resilience and optimal time complexity?

\bibliographystyle{unsrt}
\bibliography{literature}

\newpage
\appendix

\section{Proofs for Resilience Guarantees with Non-Authenticated Channels}
\subsection{Proof of \autoref{thm:consistency-robustness}}
\label{sec:proof-consistency-robustness}
\thmConsistencyRobustness*

\begin{proof}
  We first prove the claim on consistency: we assume that the number of faulty nodes $f$ is lower or equal to $\alpha \cdot n$ and that $P = H$.
  Since $|H| \ge (1-\alpha) \cdot n$, we deduce that \autoref{alg:ba} inserts at most $\lceil \frac{3}{2} \cdot (1 - \alpha) \cdot n \rceil - 1 - \lceil (1-\alpha) \cdot n \rceil$ nodes into $P$ to obtain $L$.
  We obtain:
  \begin{align*}
    \left\lceil \frac{3}{2} \cdot (1 - \alpha) \cdot n \right\rceil - 1 - \left\lceil (1-\alpha) \cdot n \right\rceil &< \frac{3}{2} \cdot (1 - \alpha) \cdot n - (1 - \alpha) \cdot n\\
    &= \frac{1}{2} \cdot (1 - \alpha) \cdot n
  \end{align*}
  where the first inequality holds since, for all real number $r$, we have $\lceil r \rceil - 1 < r$.
  This implies that $L$ contains strictly less than a third of faulty nodes.
  It is known that the Phase King algorithm reaches agreement when strictly fewer than $\lceil |L| / 3 \rceil$ nodes are faulty and the claimed consistency follows.

  We now prove the claimed robustness.
  We assume that the number of faulty nodes $f$ is lower or equal to $\lfloor \frac{1}{2} \cdot (1 - \alpha) \cdot n\rfloor - 1$, the predicted set $P$ is arbitrary.
  We now use that $|L| \ge \frac{3}{2} \cdot (1 - \alpha) \cdot n - 1$ and lower-bound the fraction of faulty nodes in $L$:
  \begin{align*}
    |F \cap L| \le f &\le \left\lfloor \frac{1}{2} \cdot (1 - \alpha) \cdot n\right\rfloor - 1\\
                     &\le \frac{1}{3} \cdot \left\lfloor \frac{3}{2} \cdot (1 - \alpha) \cdot n\right\rfloor - 1\\
                     &\le \frac{|L|}{3} - \frac{2}{3} < \frac{|L|}{3}.
  \end{align*}
  This proves that the Phase King algorithm reaches agreement in $L$ and the claim follows.
\end{proof}

\subsection{Proof of \autoref{thm:smoothness-global-predictions}}
\label{sec:proof-smoothness-global-predictions}
\thmSmoothnessGlobalPredictions*

\begin{proof}
  Let $\alpha \in [1/3, 1[$, we consider \autoref{alg:ba} with trust parameter $\alpha$.
  We note $\eta_F = |P \setminus H|$ the number of faulty nodes in $P$ and $\eta_H = |H \setminus P|$ the number of honest nodes that are not in $P$.
  Note that $\eta = \eta_F + \eta_H$.
  We will prove that the claim holds for each of the three pieces of the function $s$.

  \paragraph*{Case 1.}
  We assume that $\eta \in \llbracket 0, (1-\alpha) \cdot n - 1\rrbracket$ and that the number of faulty nodes is lower or equal to $\alpha \cdot n - \eta$.
  This implies an upper-bound on the number of honest nodes:
  \begin{align*}
    |H| \ge (1 - \alpha) \cdot n + \eta
  \end{align*}
  We now lower-bound the number of honest nodes inside $L$, $|H \cap L|$.
  First, we re-express $|H \cap L|$ in terms of $|H|$:
  \begin{align*}
    |H \cap L| = |H| - |H \setminus L|
  \end{align*}
  We note that $P \subseteq L$ by design of \autoref{alg:ba}, we therefore have $|H \setminus L| \le |H \setminus P|$.
  Combining the above two inequalities with our initial assumption that $\eta < (1-\alpha) \cdot n$ gives
  \begin{align*}
    |H \cap L|\quad  &>\quad  \eta + \eta_F \quad \ge\quad  2 \cdot \eta_F
  \end{align*}
  We obtained that the number of honest nodes inside $L$ is strictly greater than twice the number of faulty nodes $f$.
  The Phase King algorithm therefore solves the problem successfully and the claim follows.
  \\
  \paragraph*{Case 2.}
  We assume that $\eta \in  \llbracket (1-\alpha) \cdot n - 1, \frac{1+\alpha}{4} \cdot n + 1\rrbracket$.
  For this range of errors, we assume that the overall number of faulty nodes is strictly smaller than $n - 2 \cdot \eta$.
  This gives that
  \begin{align*}
    |H| > 2 \cdot \eta
  \end{align*}
  Using the identity $|H| = |H \cap P| + \eta_H$ and rearranging the terms gives
  \begin{align*}
    |H \cap P| \quad > \quad  2 \cdot \eta - \eta_H \quad  \ge \quad 2 \cdot \eta_F.
  \end{align*}
  This last inequality enables us to conclude in case $L = P$ as then, like in the previous case, strictly more than $2/3$ of the nodes of $L$ are honest which ensures the King algorithm reaches agreement.
  We now assume $P \subsetneq L$ which implies $|P| < \frac{3}{2} \cdot (1-\alpha) \cdot n - 1$.
  We use the strict lower-bound on $|H|$ obtained at the beginning of this case and obtain
  \begin{align*}
    |H \cap P| > \eta
  \end{align*}
  which, coupled with our initial assumption that $\eta \ge (1 - \alpha) \cdot n - 1$, gives
  \begin{align*}
    |H \cap P| \ge (1 - \alpha) \cdot n
  \end{align*}
  Moreover, by design of \autoref{alg:ba} it holds that 
  \begin{align*}
    |L| \le \left\lceil \frac{3}{2} \cdot (1-\alpha) \cdot n\right\rceil - 1
  \end{align*}
  gathering the two previous inequalities finally gives
  \begin{align*}
    |L| - |P| \le |L| - |H \cap P| &\le \left\lceil \frac{3}{2} \cdot (1-\alpha) \cdot n\right\rceil - 1  - (1-\alpha) \cdot n\\
    &< \frac{1}{2} \cdot (1-\alpha) \cdot n
  \end{align*}
  where the last inequality holds since, for all real number $r$, we have $\lceil r \rceil - 1 < r$.
  We deduce that, inside $L$, the number of honest nodes is strictly greater than twice the number of faulty nodes and the result follows
  \paragraph*{Case 3.}
  This case immediately holds by using the robustness of the algorithm (\autoref{thm:consistency-robustness}).\qedhere
\end{proof}

\section{Proofs for Impossibility Results with Non-Authenticated Channels}
\subsection{Proof of \autoref{thm:impossibility-consistency-robustness}}
\label{sec:proof-impossibility-consistency-robustness}

\thmImpossibilityConsistencyRobustness*

\begin{proof}
  Let $\alg$ be an algorithm with a consistency no lower than $\alpha$.
  We split the nodes into three sets $A$, $B$ and $C$ such that $|A| = |B| = \frac{1 - \alpha}{2} \cdot n$ and $|C| = \alpha \cdot n$.
  We consider the following Configuration 1: nodes in $A$ are honest with inputs $0$; nodes in $B$ are honest with inputs $1$; nodes in $C$ are faulty, they act as honest nodes with input $0$ when interacting with $A$, and as honest nodes with input $1$ when interacting with $B$.

  $\alg$ is $\alpha$-consistent, it therefore exists a prediction that makes $\alg$ reach agreement in Configuration (1).
  We now exhibit two other configurations, (2) and (3), with $\frac{1-\alpha}{2} \cdot n$ faulty nodes and show that $\alg$ fails to reach agreement in at least one of configuration (1), (2) or (3).
  We consider configurations (2) and (3) defined as follows: in Configuration 2, nodes in $A$ and $C$ are honest with inputs $0$, nodes in $B$ are faulty and run $\alg$ with input $1$; in Configuration (3), nodes in $B$ and $C$ are honest with inputs $1$, nodes in $A$ are faulty and run $\alg$ with input $0$.

  In both configurations, we give every node the prediction that caused $\alg$ to reach agreement in Configuration (1).
  By validity, all honest nodes must output $0$ in Configuration (2) (since all honest inputs are $0$) and $1$ in Configuration (3) (since all honest inputs are $1$).
  However, $A$ cannot tell (1) and (2) apart; likewise, $B$ cannot tell (1) and (3) apart.
  As $A$ and $B$ have the same outputs in (1), it follows that they also have the same outputs in (2) and (3), contradicting validity for one of (2) or (3).
  We showed that $\alg$ is either not $\alpha$-consistent or not $\frac{1 - \alpha}{2}$-robust and the claim follows.
\end{proof}

\subsection{Proof of \autoref{thm:impossibility-smoothness-appendix}}
\label{sec:proof-impossibility-smoothness-appendix}



The following proofs use \autoref{fig:smoothness-1}, \autoref{fig:smoothness-configs} and \autoref{fig:smoothness-configs-2} for visualization.
Those figures show subsets of nodes.
The orange ovals with curly borders depict Byzantine subsets while the blue ovals depict a subset of correct nodes.
Further, an arrow between a faulty subset of nodes and a honest subset of nodes means that the faulty nodes pretend to be honest nodes with input bit the label of the arrow toward the honest subset of nodes.
\thmImpossibilitySmoothnessAppendix*
\begin{proof}
  For the smoothness argument we are only interested in algorithms that have the optimal consistency-robustness tradeoff.
  Therefore, only algorithms that reach $\beta = \frac{1-\alpha}{2}\cdot n-1$ resilience for any prediction are considered.
  We first show that (1) is a set of impossible error and resilience points based on a variable choice of $\alpha \in [1/3, 1[$.
  For that we consider three configurations called (\ref{fig:smoothness-1a}), (\ref{fig:smoothness-1b}) and (\ref{fig:smoothness-1c}) depicted in \autoref{fig:smoothness-1}.
  We divide the nodes into four subsets $A$, $B$, $C$ and $D$ and a single node $x$.
  The size of $D$ is variable and the size of $C$ is determined by $|C|=\alpha \cdot n - |D| -1$. The subsets $A$ and $B$ are of equal size, i.e. $|A| = |B| = \frac{(1-\alpha)\cdot n-1}{2}$.
  As all subsets must have non-negative sizes the following additional bound applies: $|D| \leq \alpha \cdot n - 1$.
  Each configuration of \autoref{fig:smoothness-1} is defined with these five subsets.\\

  We first show that any algorithm fails in at least one of the configurations (\ref{fig:smoothness-1a}), (\ref{fig:smoothness-1b}) or (\ref{fig:smoothness-1c}) with the same input prediction.
  Let \alg be an algorithm, we assume that \alg reaches agreement in configuration (\ref{fig:smoothness-1b}) and (\ref{fig:smoothness-1c}).
  We now use a standard indistinguishability argument: the nodes in $A$ cannot distinguish between (\ref{fig:smoothness-1a}) and (\ref{fig:smoothness-1b}) while the nodes in $B$ cannot distinguish between (\ref{fig:smoothness-1a}) and (\ref{fig:smoothness-1c}).
  However, the nodes in $A$ must output $0$ in (\ref{fig:smoothness-1b}) to obey validity;
  for the same reason, the nodes in $B$ must output $1$ in (\ref{fig:smoothness-1c}).
  It follows that \alg does not reach agreement in (\ref{fig:smoothness-1a}), proving the subclaim.\\
  
  Using what we now proved we show no algorithm can have a resilience better than $\alpha \cdot n$ if it should still satisfy optimal ($\frac{1-\alpha}{2}-1$)-robustness.
  For contradiction we assume ($\alpha \cdot n + 1$)-resilience for $\eta > 0$.
  Let $\eta_a = |D| + 1$ which means the Byzantine subset $D$ and the single Byzantine node $x$ were mispredicted as honest nodes in \autoref{fig:smoothness-1a}.
  We now see the number of Byzantine nodes in (\ref{fig:smoothness-1b}) and (\ref{fig:smoothness-1c}) is $f=\frac{(1-\alpha)\cdot n-1}{2}$.
  Thus we can maximally achieve a robustness of $\beta =\frac{(1-\alpha)\cdot n-1}{2}-1$ which is strictly lower than the ($\frac{(1-\alpha)\cdot n}{2}-1$)-robustness we want to achieve.

  We now show that (2) is a set of impossible error and resilience points.
  For that, we consider three configurations called (\ref{fig:1-a}), (\ref{fig:1-b}) and (\ref{fig:1-c}) depicted in \autoref{fig:smoothness-configs}.
   We divide the nodes into four subsets called $A$, $B$, $C$ and $D$ such that $D$ is of variable size and the rest of the nodes are evenly shared among the other subsets, i.e. $|A| = |B| = |C|$.
  Each configuration of \autoref{fig:smoothness-configs} is defined with those four subsets.
  We can use the same argument as before to show that any algorithms fails in at least one of the configurations (\ref{fig:1-a}), (\ref{fig:1-b}) or (\ref{fig:1-c}) with the same input prediction.

  Using what we just proved, we now show no algorithm can have a smoothness $n - 2 \cdot \eta$.
  Let $\eta = |A|$, notice that each configuration (\ref{fig:1-a}), (\ref{fig:1-b}) and (\ref{fig:1-c}) have error $\eta$.
  The number of faulty nodes is also the same in those three configurations and equals $\eta + |D|$.
  Noticing that $|D| = n - 3 \cdot \eta$, the number of faulty nodes can be rewritten as $n - 2 \cdot \eta$, proving the claim.

  We now show that, if \alg matches (2) (that is, if it has smoothness $n - 2 \cdot \eta - 1$ for $\eta \in \llbracket n (1 - \alpha), n/3 \rrbracket$), then \alg cannot have smoothness (3).
  We define three configurations (\ref{fig:2-a}), (\ref{fig:2-b}) and (\ref{fig:2-c}) shown in \autoref{fig:smoothness-configs-2}.
  We slightly changed the sets $A'$, $B'$ and $C'$ compared to $A$, $B$ and $C$ from the above proof for (2) as follows: we took two nodes away from $C$ and inserted one in $A$ and the other in $B$.
  Assuming that \alg matches (2), \alg reaches agreement in configuration (\ref{fig:2-a}) as, after the slight changes on the subsets, this is strictly below (2).
  Notice that $A'$ must output $0$ in (\ref{fig:2-b}) but cannot distinguish between (\ref{fig:2-a}) and (\ref{fig:2-b}); likewise, $B'$ must output $1$ in (\ref{fig:2-c}) but cannot distinguish between (\ref{fig:2-a}) and (\ref{fig:2-c}), it follows that \alg does not respect validity on either (\ref{fig:2-b}) or (\ref{fig:2-c}).
  
  We now show that \alg cannot have smoothness $n/2 - 1/2 \cdot \eta - 2$ when $\eta \ge n/3$.
  For that, note that the error $\eta$ is the same in both configurations (\ref{fig:2-b}) and (\ref{fig:2-c}).
  We compute:
  \begin{align*}
    \eta &= |D| + |A|
           = |D| + \frac{n-|D|}{3} + 1
           = \frac{n}{3} + \frac{2}{3} \cdot |D| + 1
  \end{align*}
  Moreover, both configurations (\ref{fig:2-b}) and (\ref{fig:2-c}) also have the same number of faulty nodes $f = (n-|D|)/3 + 1 = n/2 - 1/2 \cdot \eta - 2$, the claim follows.
\end{proof}

\pgfmathsetmacro{\coeff}{0.66}
\begin{figure}
	\begin{subfigure}{0.32\textwidth}
		\centering
		\begin{tikzpicture}[
				edge/.style={->, draw=black, thick, -Stealth},
				scale=0.7
			]

			\node[cor={\coeff*2.0}{\coeff*1.2}] (A) at (-1.2, 2) {\small $\bm{(A)}$ : $\bm{0}$};
			\node[cor={\coeff*2.0}{\coeff*1.2}] (B) at (1.2, 2) {\small $\bm{(B)}$ : $\bm{1}$};
			\node[byz={\coeff*2.0}{\coeff*1.2}] (C) at (0, 0.5) {\small $\bm{(C)}$};
			\node[byz={\coeff*1.5}{\coeff*1.2}] (D) at (0.7, 4) {\small $\bm{(D)}$};
			\node[byz={\coeff*1}{\coeff*1}] (x) at (-0.8, 4) {\small $\bm{(x)}$};
			\begin{pgfonlayer}{background}
				\node[ellipse,fill=byznode!20, fit=(D)(x), inner sep=0.2pt, outer sep=1pt] (Dx) {};
			\end{pgfonlayer}

			\draw[edge, bend right=30] (C.east) to node [below right] {$\bm{1}$} (node cs:name=B, angle=-50);
			\draw[edge, bend left=30] (C.west) to node [below left] {$\bm{0}$} (node cs:name=A, angle=-130);

			\draw[edge, bend right=30] (Dx.west) to node [above left] {$\bm{0}$} (node cs:name=A, angle=130);
			\draw[edge, bend left=30] (Dx.east) to node [above right] {$\bm{1}$} (node cs:name=B, angle=50);

		\end{tikzpicture}
		\caption{}
		\label{fig:smoothness-1a}
	\end{subfigure}
	\begin{subfigure}{0.32\textwidth}
		\centering
		\begin{tikzpicture}[
				edge/.style={->, draw=black, thick, -Stealth},
				scale=0.7
			]

			\node[cor={\coeff*2.0}{\coeff*1.2}] (A) at (-1.2, 2) {\small $\bm{(A)}$ : $\bm{0}$};
			\node[byz={\coeff*2.0}{\coeff*1.2}] (B) at (1.2, 2) {\small $\bm{(B)}$ : $\bm{1}$};
			\node[cor={\coeff*2.0}{\coeff*1.2}] (C) at (0, 0.5) {\small $\bm{(C)}$ : $\bm{0}$};
			\node[cor={\coeff*1.5}{\coeff*1.2}] (D) at (1.2, 4) {\small $\bm{(D) : 0}$};
			\node[cor={\coeff*1}{\coeff*1}] (x) at (-1.2, 4) {\small $\bm{(x) : 0}$};

		\end{tikzpicture}
		\caption{}
		\label{fig:smoothness-1b}
	\end{subfigure}
	\begin{subfigure}{0.32\textwidth}
		\centering
		\begin{tikzpicture}[
				edge/.style={->, draw=black, thick, -Stealth},
				scale=0.7
			]

			\node[byz={\coeff*2.0}{\coeff*1.2}] (A) at (-1.2, 2) {\small $\bm{(A)}$ : $\bm{0}$};
			\node[cor={\coeff*2.0}{\coeff*1.2}] (B) at (1.2, 2) {\small $\bm{(B)}$ : $\bm{1}$};
			\node[cor={\coeff*2.0}{\coeff*1.2}] (C) at (0, 0.5) {\small $\bm{(C)}$ : $\bm{1}$};
			\node[cor={\coeff*1.5}{\coeff*1.2}] (D) at (1.2, 4) {\small $\bm{(D) : 0}$};
			\node[cor={\coeff*1}{\coeff*1}] (x) at (-1.2, 4) {\small $\bm{(x) : 0}$};

		\end{tikzpicture}
		\caption{}
		\label{fig:smoothness-1c}
	\end{subfigure}
	\caption{
		In each configuration the nodes are given the same global prediction $P = A \cup B \cup D \cup x$.
		A missing arrow between a faulty and a honest node means, the behavior of the faulty node towards that honest node is not of importance for the proof.
	}
	\label{fig:smoothness-1}
\end{figure}
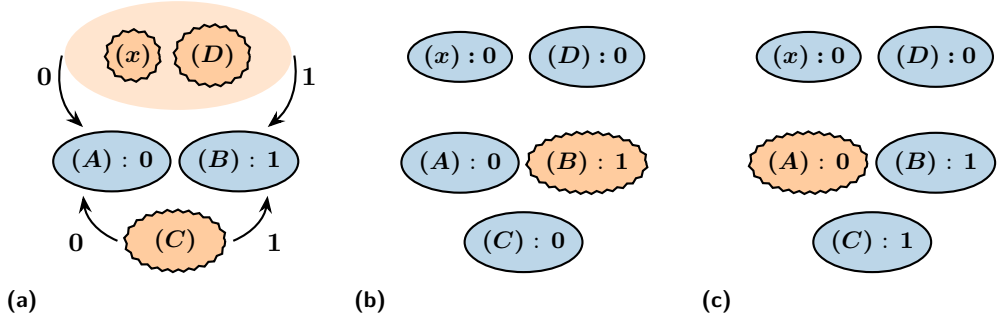

\pgfmathsetmacro{\coeff}{0.66}
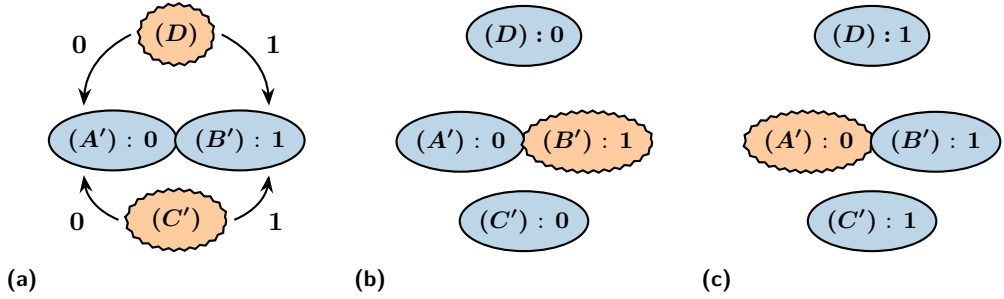
\begin{figure}
  \begin{subfigure}{0.32\textwidth}
    \centering
    \begin{tikzpicture}[
      edge/.style={->, draw=black, thick, -Stealth},
      scale=0.7
      ]
      
      \node[cor={\coeff*2.0}{\coeff*1.2}] (A) at (-1.2, 2) {\small $\bm{(A')}$ : $\bm{0}$};
      \node[cor={\coeff*2.0}{\coeff*1.2}] (B) at (1.2, 2) {\small $\bm{(B')}$ : $\bm{1}$};
      \node[byz={\coeff*2.0}{\coeff*1.2}] (C) at (0, 0.5) {\small $\bm{(C')}$};
      \node[byz={\coeff*1.5}{\coeff*1.2}] (D) at (0, 4) {\small $\bm{(D)}$};

      \draw[edge, bend right=30] (C.east) to node [below right] {$\bm{1}$} (node cs:name=B, angle=-50);
      \draw[edge, bend left=30] (C.west) to node [below left] {$\bm{0}$} (node cs:name=A, angle=-130);

      \draw[edge, bend right=30] (D.west) to node [above left] {$\bm{0}$} (node cs:name=A, angle=130);
      \draw[edge, bend left=30] (D.east) to node [above right] {$\bm{1}$} (node cs:name=B, angle=50);

    \end{tikzpicture}
    \caption{}
    \label{fig:2-a}
  \end{subfigure}
  \begin{subfigure}{0.32\textwidth}
    \centering
    \begin{tikzpicture}[
      edge/.style={->, draw=black, thick, -Stealth},
      scale=0.7
      ]
      \node[cor={\coeff*2.0}{\coeff*1.2}] (A) at (-1.2, 2) {\small $\bm{(A')}$ : $\bm{0}$};
      \node[byz={\coeff*2.0}{\coeff*1.2}] (B) at (1.2, 2) {\small $\bm{(B')}$ : $\bm{1}$};
      \node[cor={\coeff*2.0}{\coeff*1.2}] (C) at (0, 0.5) {\small $\bm{(C')}$ : $\bm{0}$};
      \node[cor={\coeff*1.5}{\coeff*1.2}] (D) at (0, 4) {\small $\bm{(D) : 0}$};

    \end{tikzpicture}
    \caption{}
    \label{fig:2-b}
  \end{subfigure}
  \begin{subfigure}{0.32\textwidth}
    \centering
    \begin{tikzpicture}[
      edge/.style={->, draw=black, thick, -Stealth},
      scale=0.7
      ]
      
      \node[byz={\coeff*2.0}{\coeff*1.2}] (A) at (-1.2, 2) {\small $\bm{(A')}$ : $\bm{0}$};
      \node[cor={\coeff*2.0}{\coeff*1.2}] (B) at (1.2, 2) {\small $\bm{(B')}$ : $\bm{1}$};
      \node[cor={\coeff*2.0}{\coeff*1.2}] (C) at (0, 0.5) {\small $\bm{(C')}$ : $\bm{1}$};
      \node[cor={\coeff*1.5}{\coeff*1.2}] (D) at (0, 4) {\small $\bm{(D) : 1}$};

    \end{tikzpicture}
    \caption{}
    \label{fig:2-c}
  \end{subfigure}
  \caption{
    In each configuration the nodes are given the same global prediction $P = A' \cup B' \cup C'$.
    Compared to \autoref{fig:smoothness-configs}, the set $C'$ has two fewer nodes than $C$ and the sets $A'$ and $B'$ each have one more node than respectively $A$ and $B$.
  }
  \label{fig:smoothness-configs-2}
\end{figure}

\section{Reaching Agreement with Authenticated Channels}
\label{sec:auth}

\begin{algorithm}[H]
  \caption{\localg($v_x$, $P$, $\alpha$)}\label{alg:auth-ba}
  \tcp{This algorithm runs in node $x$}
  \KwIn{the input bit $v_x$, prediction set $P \subseteq \llbracket n\rrbracket$, trust parameter $\alpha \in [1/2,1]$}

  $L \gets P$\\

  $y \gets 1$\\
  \If{$|L| < 2 \cdot (1 - \alpha) \cdot  n - 1$}{
    \Repeat{$|L| \ge 2 \cdot (1 - \alpha) \cdot  n - 1$}{
      \If{$y \notin L$}{
        $L \gets L \cup \{y\}$\\
        $y \gets y + 1$\\
      }
    }
  }
  $t \gets \lceil|L| / 2\rceil$\\
  \eIf{$x \in L$}{
    Run the \authAlg algorithm among the nodes in $L$.\\
    Once \authAlg terminated, broadcast the output.\\
    \KwOut{the output of the \authAlg algorithm}
  }
  {
    Wait the number of rounds taken by \authAlg with $|L|$ nodes and $t-1$ faulty nodes to terminate.\\
    \KwOut{a value received at least $|L|-t+1$ times, otherwise output the initial input bit}
    
  }
\end{algorithm}



\begin{restatable}{theorem}{thmAuthConsistencyRobustness}
  \label{thm:auth-consistency-robustness}
  Let $\alpha \in [1/2, 1]$ and $n$ be the number of nodes.
  \autoref{alg:auth-ba} with input trust parameter $\alpha$ is $\alpha$-consistent and $(1 - \alpha - 1/n)$-robust.
\end{restatable}

\begin{proof}
  Let $\alpha \in [1/2, 1]$.
  We first show the claim about consistency.
  We assume that the number of faulty nodes $f = |F|$ is lower or equal to $\alpha \cdot n$ and that the prediction is perfect, i.e., $P = H$.
  We have that $|P| = |H| \ge (1-\alpha) \cdot n$, hence \autoref{alg:auth-ba} inserts strictly fewer than $(1 - \alpha) \cdot n$ nodes in $P$ to obtain $L$.
  It follows that strictly more than half of the nodes in $L$ are honest, hence the \authAlg algorithm reaches agreement and the claim on consistency follows.

  Now on robustness: we assume that there are at most $(1-\alpha) \cdot n - 1$ faulty nodes, the predicted set $P$ is arbitrary.
  As the set $L$ has size no lower than $\lceil 2 (1 - \alpha) \cdot n \rceil - 1$, we deduce that strictly more than half of the nodes of $L$ are honest and the claim on robustness follows.
\end{proof}


\begin{restatable}[Smoothness]{theorem}{thmAuthSmoothnessGlobalPredictions}
  \label{thm:auth-smoothness-global-predictions}
  Let $\alpha \in [1/2, 1]$.
  We consider the function $s$ defined as follows.
  \[
    s(\eta) = \begin{cases}
      \alpha \cdot n - \frac{1}{2} \cdot \eta & \text{if } \eta \in \llbracket 0, 2 (1-\alpha) \cdot n\rrbracket, \\[4pt]
      n - \frac{3}{2} \cdot \eta - 1 & \text{if } \eta \in \llbracket 2 (1-\alpha) \cdot n, \frac{2 \alpha}{3} \cdot n \rrbracket, \\[4pt]
      (1-\alpha) \cdot n - 1 &\text{if } \eta \in \llbracket \frac{2 \alpha}{3} \cdot n, n \rrbracket
    \end{cases}
  \]
  Assuming an error of $\eta$, \autoref{alg:auth-ba} with trust parameter $\alpha$ reaches agreement against $s(\eta)$ faulty nodes.
\end{restatable}

\begin{proof}
  Let $\alpha \in [1/2, 1]$, we consider \autoref{alg:auth-ba} with trust parameter $\alpha$.
  We call $\eta_F$ the number of faulty nodes in $P$ and $\eta_H$ the number of honest nodes not in $P$.
  By definition, it holds that $\eta = \eta_F + \eta_H$.\\
  We divide our proof in three parts, one part for each piece of the function $s$ definition: we show that \autoref{alg:auth-ba} has smoothness $s$ when (1) $\eta \in [0, 2(1-\alpha) \cdot n[$, then (2) when $\eta \in [2(1-\alpha) \cdot n, \frac{2\alpha }{3} \cdot n]$ and finally (3) when $\eta \in [\frac{2\alpha }{3} \cdot n, n]$.
  \\\\
  \noindent\textbf{Case 1.}
    In this case, we assume that $\eta \in [0, 2 (1-\alpha) \cdot n]$.
    We assume that the number of faulty nodes $|F|$ is strictly smaller than $\alpha \cdot n - \eta / 2$.
    This directly implies an upper-bound on the number of honest nodes:
    \begin{align*}
      |H| > (1 - \alpha) \cdot n + \eta / 2.
    \end{align*}
    Recall that $L$ is the set of trusted nodes of the algorithm, that is, the set of nodes in which we run the \authAlg algorithm.
    \autoref{alg:ba} creates $L$ by inserting more nodes in $P$ in case $|P| < \lceil n \cdot 2 (1-\alpha)\rceil - 1$, otherwise $L = P$.
    In either cases, there are at most $\eta_H$ honest nodes outside of $L$.
    We lower-bound the number of honest nodes in $L$ using the fact that $\eta_H \ge |H \setminus L|$ and that $\eta = \eta_F + \eta_H$:
    \begin{align*}
      |H| = |H \setminus L| + |L \cap H| &> (1 - \alpha) \cdot n + \eta/2\\
      \implies \eta_H + |L \cap H| &> (1 - \alpha) \cdot n + \eta/2
    \end{align*}
    We then use our initial assumption that $\eta \le 2(1-\alpha) \cdot n$ and derive
    \begin{align*}
      \eta_H + |L \cap H| > \eta/2 + \eta/2 = \eta\quad
      \implies\quad |L \cap H| > \eta_F
    \end{align*}
    Hence, strictly more than half of the nodes in $L$ are honest, the \authAlg algorithm therefore solves the problem successfully in $L$ and \autoref{alg:auth-ba} reaches agreement, the claim follows.
    \\\\
  \noindent\textbf{Case 2.}
    In this case, we assume that $\eta \in  [2(1-\alpha) \cdot n, \frac{2\alpha }{3} \cdot n]$ and that the number of faulty nodes is strictly smaller than $n - 3/2 \cdot \eta$.
    We then have that
    \begin{align*}
      |F| < n - \frac{3}{2} \cdot \eta\quad
      \implies\quad |H| > \frac{3}{2} \cdot \eta
    \end{align*}
    Noting that $\eta_H = |H \setminus P|$, it holds that
    \begin{align*}
      |H| = |H \setminus P| + |H \cap P| > \frac{3}{2} \cdot \eta\quad
      \implies\quad |H \cap P| > \frac{3}{2} \cdot \eta_F
    \end{align*}
    We can already conclude in case $L = P$, that is, in case $|P| \ge \lceil n \cdot 2 (1-\alpha)\rceil - 1$: then strictly more than half of the nodes in $L$ are honest and the claim follows as in the previous case.\\
    We now assume that $|P| < \lceil n \cdot 2 (1-\alpha)\rceil - 1$.
    Using the equality our assumption that $|H| > \eta$ and the obvious fact that $\eta_H \le \eta$ gives
    \begin{align*}
      |H \setminus P| + |P \cap H| > \frac{3}{2} \cdot \eta\quad
      \implies\quad |P \cap H| > \frac{1}{2} \cdot \eta
    \end{align*}
    We then use our assumption that $\eta \ge 2 (1-\alpha) \cdot n$ and obtain
    \begin{align*}
      |P \cap H| &> (1 - \alpha) \cdot n
    \end{align*}
    Hence $P$ already contains $(1 - \alpha) \cdot n$ honest nodes.
    As $L$ has a size no greater than $2 (1 - \alpha) \cdot n$, more than half of the nodes in $L$ are therefore honest and the claim follows.
    \\\\
  \noindent\textbf{Case 3.}
    We assume that $\eta \in  \llbracket \frac{2\alpha }{3} \cdot n, n \rrbracket$.
    We use the robustness of \autoref{alg:auth-ba} (\autoref{thm:auth-consistency-robustness}) and the claim immediately follows.\qedhere
\end{proof}



\begin{restatable}[Impossibility for Authenticated Consistency and Robustness]{theorem}{thmAuthImpossibilityConsistencyRobustness}
  \label{thm:auth-impossibility-consistency-robustness}
  Let $\alpha \in [1/2, 1]$.
  In the authenticated framework, any prediction-augmented algorithm with consistency $\alpha$ has a robustness strictly lower than $1 - \alpha$.
  This result is independent from the chosen prediction system.
\end{restatable}

\begin{proof}
Let \alg be an algorithm that is $\alpha$-consistent.
  We split the nodes in two subsets $A$ and $B$ such that $|A| = \alpha \cdot n$ and $|B| = (1 - \alpha) \cdot n$.
  We consider three configuration, Configuration 1, 2 and 3.
  
  In Configuration 1, the nodes of $A$ are honest with input $0$ while the nodes of $B$ are faulty; the nodes receive a prediction $P$ such that \alg reaches agreement ($\alpha$-consistency), the nodes of $B$ play as if they had input $1$ and received prediction $P$.
  In Configuration 2, the nodes of $A$ are faulty and the nodes of $B$ are honest with input $1$; the nodes receive the same prediction $P$ as in Configuration 1, the nodes of $A$ play as if they had input $0$ and received prediction $P$.
  In Configuration 3, finally, all the nodes are honest and receive prediction $P$, the nodes of $A$ have input $0$ and the nodes of $B$ have input $1$.
  
  We see that $A$ cannot distinguish between Configurations 1 and 3; similarly, $B$ cannot distinguish between Configurations 2 and 3.
  $A$ must output $0$ in Configuration 1 and $B$ must output $1$ in Configuration 2 which prevents from reaching agreement in Configuration 3.
  \alg fails to solve the problem in either Configuration 2 or Configuration 3 which shows \alg is not $(1-\alpha)$-robust, proving the claim.
\end{proof}



\begin{restatable}[Impossibility on Smoothness]{theorem}{thmAuthImpossibilitySmoothness}
  \label{thm:auth-impossibility-smoothness}
  Let $\alpha \in [1/2,1]$, let \alg be an algorithm that is $\alpha$-consistent. We consider the function $\bar{s}$ defined as follows.
	\[
		\bar{s}(\eta) = \left\{\begin{array}{l@{\quad}l@{\quad}r}
			\alpha \cdot n + 1 & if \eta \in \llbracket 0, (1-\alpha)\cdot n \rrbracket, & (1) \\
			n -  \eta & \text{if } \eta \in \llbracket (1-\alpha) \cdot n, \alpha \cdot n\rrbracket, & (2) \\
		\end{array}\right.
	\]
  The smoothness function of \alg is strictly lower than $\bar{s}$.
\end{restatable}

\begin{proof}
  Let \alg be an algorithm that is $\alpha$-consistent and should satisfy optimal ($1-\alpha-\frac{1}{n}$)-robustness. We first show that (1) is a set of impossible error and resilience points based on a variable choice of $\alpha \in [1/2,1[$. For that we consider three configurations called (\ref{fig:auth_smoothness-1a}), (\ref{fig:auth_smoothness-1b}) and (\ref{fig:auth_smoothness-1c}) depicted in \autoref{fig:auth_smoothness-1}.
  We divide the nodes into three subsets $A$, $B$ and $C$ and a single node $x$.
  The size of $C$ is variable and the size of $B$ is determined by $|B| = \alpha \cdot n - C$. The subset A is of size $(1-\alpha)\cdot n -1$.
  As all subsets must have non-negative sizes the size of $C$ is limited by $|C| \leq \alpha \cdot n -1$.

  We first show that any algorithm fails in at least one of the configurations (\ref{fig:auth_smoothness-1a}), (\ref{fig:auth_smoothness-1b}) and (\ref{fig:auth_smoothness-1c}) with the same input prediction.
  We see that $A$ cannot distinguish between configurations (\ref{fig:auth_smoothness-1b}) and (\ref{fig:auth_smoothness-1c}); similarly $B$ cannot distinguish between configurations (\ref{fig:auth_smoothness-1b}) and (\ref{fig:auth_smoothness-1c}). $A$ must output $0$ in configuration (\ref{fig:auth_smoothness-1a}) and $B$ must output $1$ in configuration (\ref{fig:auth_smoothness-1b}) which prevents from reaching agreement in configuration (\ref{fig:auth_smoothness-1c}).

  Using what we now proved we show no algorithm can have a resilience better that $\alpha \cdot n$ if it should still satisfy optimal ($1-\alpha-\frac{\eta}{n}$)-robustness.
  For contradiction we assume ($\alpha \cdot n + 1$)-resilience for $\eta > 0$. Let $\eta_a = |C|+1$ which can be understood as a misprediction of the Byzantine subset $C$ and the single node $x$ in configuration (\ref{fig:auth_smoothness-1a}).
  We now see the number of Byzantine nodes in (\ref{fig:auth_smoothness-1b}) is $f = (1-\alpha)\cdot n -1$. Thus we can macimally achieve a robustness of $\beta = (1-\alpha)\cdot n -2$ which is strictly lower thatn the ($(1-\alpha)\cdot n-1$)-robustness we want to achieve.
  From the limits on $|C|$ we can also conclude $\eta_a \leq \alpha \cdot n$.
  
  We now show that (2) is a set of impossible error and resilience points.
  For that we consider three configurations (a), (b) and (c).
  We divide the nodes into three subsets $A$, $B$ and $C$ such that $|C| = c$ and the rest of the nodes are evenly shared between $A$ and $B$, i.e., $|A| = |B|$.
  Let the size of $C$ be variable with $|C| \in [0, \alpha \cdot n]$
  Each configuration is slightly different and defined as follows.
  In configuration $a$, the nodes of $C$ are faulty, the nodes of $A$ are honest with input $0$, the nodes of $B$ are honest with input $1$, the nodes are given as input a prediction $P$ that make \alg reach agreement in this configuration ($\alpha$-consistency of \alg); the nodes of $C$ are playing as if they had input $0$ and prediction $P$ with the nodes of $A$, and as if they had input $1$ and prediction $P$ with the nodes of $B$.
  In configuration $b$, the nodes of $B$ are faulty and play as if they had input $1$ with prediction $P$, the other nodes have input $0$ and prediction $P$.
  In configuration $c$, the nodes of $A$ are faulty and play as if they had input $0$ with prediction $P$, the other nodes have input $1$ and prediction $P$.
  
  \alg reaches agreement in configuration $a$ as it is $\alpha$-consistent.
  However, the nodes of $A$ cannot distinguish between configurations $a$ and $b$, likewise, the nodes of $B$ cannot distinguish between configurations $a$ and $c$: \alg does not solve the problem for one of the configurations $a$ and $b$.
  
  We now show our claim.
  Notice that configuration $a$ and $b$ both have an error of $\eta = (n+|C|)/2$ and a number of faulty nodes $f = (n-|C|)/2$.
  We showed that \alg cannot solve the problem against $n - \eta$ faulty nodes, proving the claim.
\end{proof}
\pgfmathsetmacro{\coeff}{0.66}
\begin{figure}
	\begin{subfigure}{0.32\textwidth}
		\centering
		\begin{tikzpicture}[
				edge/.style={->, draw=black, thick, -Stealth},
				scale=0.7
			]

			\node[cor={\coeff*2.0}{\coeff*1.2}] (A) at (-1.2, 2) {\small $\bm{(A)}$ : $\bm{0}$};
			\node[byz={\coeff*2.0}{\coeff*1.2}] (B) at (1.2, 2) {\small $\bm{(B)}$ : $\bm{1}$};
			\node[byz={\coeff*1.5}{\coeff*1.2}] (C) at (1.2, 3.5) {\small $\bm{(C)}$ : $\bm{1}$};
			\node[byz={\coeff*1}{\coeff*1}] (x) at (-1.2, 3.5) {\small $\bm{(x)}$ : $\bm{1}$};

		\end{tikzpicture}
		\caption{}
		\label{fig:auth_smoothness-1a}
	\end{subfigure}
	\begin{subfigure}{0.32\textwidth}
		\centering
		\begin{tikzpicture}[
				edge/.style={->, draw=black, thick, -Stealth},
				scale=0.7
			]

			\node[byz={\coeff*2.0}{\coeff*1.2}] (A) at (-1.2, 2) {\small $\bm{(A)}$ : $\bm{0}$};
			\node[cor={\coeff*2.0}{\coeff*1.2}] (B) at (1.2, 2) {\small $\bm{(B)}$ : $\bm{1}$};
			\node[cor={\coeff*1.5}{\coeff*1.2}] (C) at (1.2, 3.5) {\small $\bm{(C)}$ : $\bm{1}$};
			\node[cor={\coeff*1}{\coeff*1}] (x) at (-1.2, 3.5) {\small $\bm{(x)}$ : $\bm{1}$};

		\end{tikzpicture}
		\caption{}
		\label{fig:auth_smoothness-1b}
	\end{subfigure}
	\begin{subfigure}{0.32\textwidth}
		\centering
		\begin{tikzpicture}[
				edge/.style={->, draw=black, thick, -Stealth},
				scale=0.7
			]

			\node[cor={\coeff*2.0}{\coeff*1.2}] (A) at (-1.2, 2) {\small $\bm{(A)}$ : $\bm{0}$};
			\node[cor={\coeff*2.0}{\coeff*1.2}] (B) at (1.2, 2) {\small $\bm{(B)}$ : $\bm{1}$};
			\node[cor={\coeff*1.5}{\coeff*1.2}] (C) at (1.2, 3.5) {\small $\bm{(C)}$ : $\bm{1}$};
			\node[cor={\coeff*1}{\coeff*1}] (x) at (-1.2, 3.5) {\small $\bm{(x)}$ : $\bm{1}$};

		\end{tikzpicture}
		\caption{}
		\label{fig:auth_smoothness-1c}
	\end{subfigure}
	\caption{
		In each configuration the nodes are given the same global prediction $P = A \cup C \cup x$.
	}
	\label{fig:auth_smoothness-1}
\end{figure}
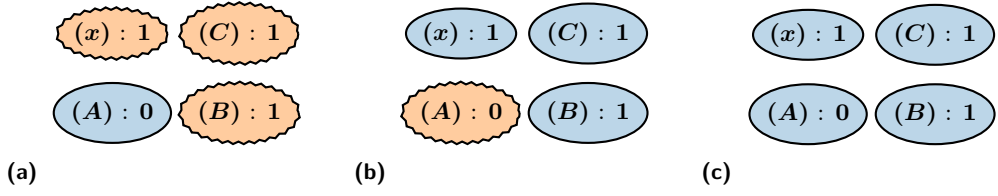

\section{Proof of \autoref{thm:impossibility-local}}
\label{sec:proof-impossibility-local}

\thmImpossibilityLocal*

\pgfmathsetmacro{\coeff}{0.66}
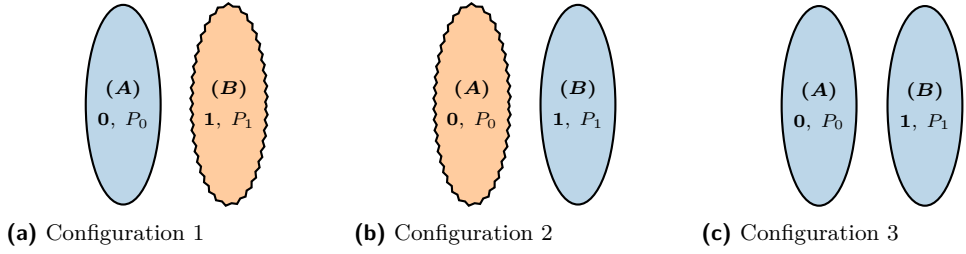
\begin{figure}
  \begin{subfigure}{0.32\textwidth}
    \centering
    \begin{tikzpicture}[
      edge/.style={->, draw=black, thick, -Stealth},
      scale=0.7
      ]
      \node[cor={\coeff*1.5}{\coeff*4}] (A) at (-1, 0) {$\substack{\bm{(A)}\\ \\ \bm{0}, \ P_0}$};
      \node[byz={\coeff*1.5}{\coeff*4}] (B) at (1, 0) {$\substack{\bm{(B)}\\ \\ \bm{1}, \ P_1}$};
    \end{tikzpicture}
    \caption{Configuration 1}
  \end{subfigure}
  \begin{subfigure}{0.32\textwidth}
    \centering
    \begin{tikzpicture}[
      edge/.style={->, draw=black, thick, -Stealth},
      scale=0.7
      ]
      \node[byz={\coeff*1.5}{\coeff*4}] (A) at (-1, 0) {$\substack{\bm{(A)}\\ \\ \bm{0}, \ P_0}$};
      \node[cor={\coeff*1.5}{\coeff*4}] (B) at (1, 0) {$\substack{\bm{(B)}\\ \\ \bm{1}, \ P_1}$};
    \end{tikzpicture}
    \caption{Configuration 2}
  \end{subfigure}
  \begin{subfigure}{0.32\textwidth}
    \centering
    \begin{tikzpicture}[
      edge/.style={->, draw=black, thick, -Stealth},
      scale=0.7
      ]
      \node[cor={\coeff*1.5}{\coeff*4}] (A) at (-1, 0) {$\substack{\bm{(A)}\\ \\ \bm{0}, \ P_0}$};
      \node[cor={\coeff*1.5}{\coeff*4}] (B) at (1, 0) {$\substack{\bm{(B)}\\ \\ \bm{1}, \ P_1}$};
    \end{tikzpicture}
    \caption{Configuration 3}
  \end{subfigure}
  \caption{
    Different indistinguishable configurations where $P_0$ is a consistency prediction for Configuration 1 and $P_1$ is a consistency prediction for Configuration 2.
    No Algorithm that reaches agreement in Configurations 1 and 2 can reach agreement in Configuration 3.
  }
  \label{fig:local-predictions}
\end{figure}

\begin{proof}
  Let $\alpha \in [1/2, 1[$, we split the nodes in two sets $A$ and $B$ of equal size assuming $n$ is divisible by two.
  We consider an algorithm \alg which is $\alpha$-consistent.
  Notice that, in particular, \alg is $1/2$-consistent.\\

  We consider the following predictions, (see \autoref{fig:local-predictions}):
  \begin{itemize}
  \item
    $P_0$ is a prediction such that \alg reaches agreement when the nodes of $B$ are faulty and the nodes of $A$ are honest and have input bit $0$.
  \item
    $P_1$ is a prediction such that \alg reaches agreement when the nodes of $A$ are faulty and the nodes of $B$ are honest and have input bit $1$.
  \end{itemize}
  
  We now consider the following configurations:\\
  \textbf{Configuration 1:}
  \begin{itemize}
  \item
    Nodes in $A$ are honest with input bit $0$ and prediction $P_0$.
  \item
    Nodes in $B$ are faulty nodes and play as if they had input bit $1$ and prediction $P_1$.
  \end{itemize}
  \textbf{Configuration 2:}
  \begin{itemize}
  \item
    Nodes in $A$ are faulty nodes and play as if they had input bit $0$ and prediction $P_0$.
  \item
    Nodes in $B$ are honest with input bit $1$ and prediction $P_1$.
  \end{itemize}
  \textbf{Configuration 3:}
  \begin{itemize}
  \item
    Nodes in $A$ are honest with input bit $0$ and prediction $P_0$.
  \item
    Nodes in $B$ are honest with input bit $1$ and prediction $P_1$.
  \end{itemize}

  As we assumed \alg is $\alpha$-consistent, the nodes of $A$ agree on $0$ in Configuration 1 while the nodes of $B$ agree on $1$ on configuration $2$.
  The nodes of $A$ however cannot distinguish between configurations 1 and 3, the same goes for the nodes of $B$ between configurations 2 and 3.
  The nodes of $A$ and $B$ therefore have different output in Configuration 3 even though there are no faulty nodes, \alg is therefore $0$-robust.
\end{proof}
\end{document}